\documentclass[11pt]{article}
\pdfoutput=1
\usepackage[protrusion=true,expansion=true]{microtype}
\usepackage[dvipsnames]{xcolor}
\usepackage{amsmath,amssymb,amsfonts,amsthm}
\usepackage{mathtools}
\usepackage{subcaption}
\usepackage{float}
\usepackage{graphicx}
\usepackage[backref=page]{hyperref}
\usepackage{color}
\usepackage{wrapfig}
\usepackage{tikz}
\usetikzlibrary{decorations.pathreplacing}
\usepackage{setspace}
\usepackage{algorithm}
\usepackage[noend]{algpseudocode}
\usepackage[framemethod=tikz]{mdframed}
\usepackage{xspace}
\usepackage{pgfplots}
\usepackage{framed}
\usepackage{subcaption}
\usepackage{thmtools}
\usepackage{thm-restate}
\usepackage{comment}
\usepackage{bm}

\pgfplotsset{compat=1.5}

\newtheorem{theorem}{Theorem}[section]

\newtheorem{lemma}[theorem]{Lemma}

\newtheorem{definition}[theorem]{Definition}

\newenvironment{proofof}[1]{\begin{trivlist} \item {\bf Proof
#1:~~}}
  {\qed\end{trivlist}}
\renewenvironment{proofof}[1]{\par\medskip\noindent{\bf Proof of #1: \ }}{\hfill$\Box$\par\medskip}
\newcommand{\namedref}[2]{\hyperref[#2]{#1~\ref*{#2}}}
\newcommand{\thmlab}[1]{\label{thm:#1}}
\newcommand{\thmref}[1]{\namedref{Theorem}{thm:#1}}
\newcommand{\lemlab}[1]{\label{lem:#1}}
\newcommand{\lemref}[1]{\namedref{Lemma}{lem:#1}}

\newcommand{\seclab}[1]{\label{sec:#1}}
\newcommand{\secref}[1]{\namedref{Section}{sec:#1}}

\newcommand{\figlab}[1]{\label{fig:#1}}
\newcommand{\figref}[1]{\namedref{Figure}{fig:#1}}
\newcommand{\alglab}[1]{\label{alg:#1}}
\renewcommand{\algref}[1]{\namedref{Algorithm}{alg:#1}}

\def \privmed    {\mdef{\mathsf{PrivMed}}}

\def \adversary    {\mdef{\mathcal{A}}}
\def \honestalg    {\mdef{\mathsf{HonestAlg}}}

\newcommand\norm[1]{\left\lVert#1\right\rVert}
\newcommand{\PPr}[1]{\ensuremath{\mathbf{Pr}\left[#1\right]}}
\newcommand{\PPPr}[2]{\ensuremath{\underset{#1}{\mathbf{Pr}}\left[#2\right]}}
\newcommand{\Ex}[1]{\ensuremath{\mathbb{E}\left[#1\right]}}
\newcommand{\EEx}[2]{\ensuremath{\underset{#1}{\mathbb{E}}\left[#2\right]}}
\renewcommand{\O}[1]{\ensuremath{\mathcal{O}\left(#1\right)}}
\newcommand{\tO}[1]{\ensuremath{\widetilde{\mathcal{O}}\left(#1\right)}}
\newcommand{\eps}{\varepsilon}

\def \ba    {\mdef{\mathbf{a}}}
\def \calA    {\mdef{\mathcal{A}}}

\def \calD    {\mdef{\mathcal{D}}}

\def \calN    {\mdef{\mathcal{N}}}

\def \calT    {\mdef{\mathcal{T}}}

\def \calV    {\mdef{\mathcal{V}}}

\def \calX    {\mdef{\mathcal{X}}}
\def \calY    {\mdef{\mathcal{Y}}}

\def \bA    {\mdef{\mathbf{A}}}

\def \bG    {\mdef{\mathbf{G}}}

\def \bM    {\mdef{\mathbf{M}}}

\def \bR    {\mdef{\mathbf{R}}}
\def \bS    {\mdef{\mathbf{S}}}

\def \bV    {\mdef{\mathbf{V}}}
\def \bU    {\mdef{\mathbf{U}}}
\def \bX    {\mdef{\mathbf{X}}}

\def \bb    {\mdef{\mathbf{b}}}
\def \be    {\mdef{\mathbf{e}}}

\def \bq    {\mdef{\mathbf{q}}}

\def \bu    {\mdef{\mathbf{u}}}

\def \bv    {\mdef{\mathbf{v}}}
\def \bx    {\mdef{\mathbf{x}}}
\def \by    {\mdef{\mathbf{y}}}
\def \bz    {\mdef{\mathbf{z}}}

\newcommand{\mdef}[1]{{\ensuremath{#1}}\xspace}  

\DeclareMathOperator*{\poly}{poly}

\DeclareMathOperator*{\kde}{KDE}

\DeclareMathOperator*{\nnz}{nnz}

\DeclareMathOperator*{\issucc}{success}

\newcommand{\superscript}[1]{\ensuremath{^{\mbox{\tiny{\textit{#1}}}}}\xspace}
\def \nd {\superscript{nd}}     

\newcommand{\abs}[1]{\mdef{\left|#1\right|}}         

\newcommand{\ignore}[1]{}

\newif\ifnotes\notestrue 
\ifnotes
\newcommand{\samson}[1]{\textcolor{purple}{{\bf (Samson:} {#1}{\bf ) }} \marginpar{\tiny\bf
             \begin{minipage}[t]{0.5in}
               \raggedright S:
            \end{minipage}}}            							
\else
\newcommand{\samson}[1]{}
\fi

\hypersetup{
     colorlinks   = true,
     citecolor    = blue,
		 linkcolor		= red
}

\makeatletter
\renewcommand*{\@fnsymbol}[1]{\textcolor{mahogany}{\ensuremath{\ifcase#1\or *\or \dagger\or \ddagger\or
 \mathsection\or \triangledown\or \mathparagraph\or \|\or **\or \dagger\dagger
   \or \ddagger\ddagger \else\@ctrerr\fi}}}
\makeatother

\providecommand{\email}[1]{\href{mailto:#1}{\nolinkurl{#1}\xspace}}

\definecolor{mahogany}{rgb}{0.75, 0.25, 0.0}
\definecolor{bleudefrance}{rgb}{0.19, 0.55, 0.91}
\hypersetup{
     colorlinks   = true,
     citecolor    = bleudefrance,
  	 linkcolor	  = bleudefrance
}

\DeclarePairedDelimiterX{\inp}[2]{\langle}{\rangle}{#1, #2}

\DeclarePairedDelimiterX{\infdivx}[2]{(}{)}{%
  #1\;\delimsize\|\;#2%
}

\newcommand{\mc}[1]{\mathcal{#1}}

\newcommand{\R}{\mathbb{R}}

\newcommand{\N}{\mathbb{N}}

\renewcommand{\P}{\mathbb{P}}

\newcommand{\lprp}[1]{\left(#1\right)}
\newcommand{\lbrb}[1]{\left\{#1\right\}}

\newcommand{\wt}[1]{\widetilde{#1}}

\newcommand{\unif}{\mathrm{Unif}}

\newcommand{\quant}{\mathrm{Quant}}
\newcommand{\srhtadealg}{\mathrm{RetNorm}}

\def\rows{n}

\def \bb    {\mdef{\mathbf{b}}}
\def \be    {\mdef{\mathbf{e}}}
\numberwithin{equation}{section}

\usepackage{cleveref} 
\def \bSig    {\mdef{\mathbf{\Sigma}}}

\usepackage{fullpage}
\begin{document}

\title{Robust Algorithms on Adaptive Inputs from Bounded Adversaries}
\author{Yeshwanth Cherapanamjeri\thanks{UC Berkeley. 
E-mail: \email{yeshwanth@berkeley.edu}}
\and
Sandeep Silwal\thanks{MIT. 
E-mail: \email{silwal@mit.edu}
Sandeep Silwal is supported by an NSF Graduate Research Fellowship under Grant No.\ 1745302, and NSF TRIPODS program (award DMS-2022448), NSF award CCF-2006798, and Simons Investigator Award (via Piotr Indyk).}
\and
David P. Woodruff\thanks{Carnegie Mellon University. 
E-mail: \email{dwoodruf@andrew.cmu.edu}. 
Work done in part while at Google Research. Partially supported by a Simons Investigator Award and by the National Science Foundation under Grant No. CCF-1815840.}
\and
Fred Zhang\thanks{UC Berkeley. 
E-mail: \email{z0@berkeley.edu}. 
Supported by ONR grant N00014-18-1-2562. Part of work done while interning at Google.}
\and
Qiuyi (Richard) Zhang\thanks{Google Research. E-mail: \email{qiuyiz@google.com}.}
\and
Samson Zhou\thanks{UC Berkeley and Rice University. 
E-mail: \email{samsonzhou@gmail.com}. 
Work done in part while at Carnegie Mellon University. Partially supported by a Simons Investigator Award and by the National Science Foundation under Grant No. CCF-1815840.}
}
\date{\today}

\maketitle

\begin{abstract}
We study dynamic algorithms robust to adaptive input generated from sources with bounded capabilities, such as sparsity or limited interaction. For example, we consider robust linear algebraic algorithms when the updates to the input are sparse but given by an adversary with access to a query oracle. We also study robust algorithms in the standard centralized setting, where an adversary queries an algorithm in an adaptive manner, but the number of interactions between the adversary and the algorithm is bounded. We first recall a unified framework of \cite{HassidimKMMS20,BeimelKMNSS22,AttiasCSS23} for answering $Q$ adaptive queries that incurs $\widetilde{\mathcal{O}}(\sqrt{Q})$ overhead in space, which is roughly a quadratic improvement over the na\"{i}ve implementation, and only incurs a logarithmic overhead in query time. Although the general framework has diverse applications in machine learning and data science, such as adaptive distance estimation, kernel density estimation, linear regression, range queries, and point queries and serves as a preliminary benchmark, we demonstrate even better algorithmic improvements for (1) reducing the pre-processing time for adaptive distance estimation and (2) permitting an unlimited number of adaptive queries for kernel density estimation. Finally, we complement our theoretical results with additional empirical evaluations. 
\end{abstract}
\newpage
\section{Introduction}
Robustness to adaptive inputs or adversarial attacks has recently emerged as an important desirable characteristic for algorithm design. 
An adversarial input can be created using knowledge of the model to induce incorrect outputs on widely used models, such as neural networks~\cite{BiggioCMNSLGR13,SzegedyZSBEGF13,GoodfellowSS15,CarliniW17,MadryMSTV18}. 
Adversarial attacks against machine learning algorithms in practice have also been documented in applications such as network monitoring~\cite{ChandolaBK09}, strategic classification~\cite{HardtMPW16}, and autonomous navigation~\cite{PapernotMG16,LiuCLS17,PapernotMGJCS17}. 
The need for sound theoretical understanding of adversarial robustness is also salient in situations where successive inputs to an algorithm can be possibly correlated; even if the input is not adversarially generated, a user may need to repeatedly interact with a mechanism in a way such that future updates may depend on the outcomes of previous interactions~\cite{MironovNS11,GilbertHSWW12,BogunovicMSC17,NaorY19,AvdiukhinMYZ19}. 
Motivated by both practical needs and a lack of theoretical understanding, there has been a recent flurry of theoretical studies of adversarial robustness. 
The streaming model of computation has especially received significant attention~\cite{Ben-EliezerJWY21,HassidimKMMS20,WoodruffZ21,KaplanMNS21,BravermanHMSSZ21,ChakrabartiGS22,AjtaiBJSSWZ22,ChakrabartiGS22,Ben-EliezerEO22,AssadiCGS22,AttiasCSS23,DinurSWZ23,WoodruffZZ23}. 
More recently, there have also been a few initial results for dynamic algorithms on adaptive inputs for graph algorithms~\cite{Wajc20,BeimelKMNSS22,BernsteinBGNSS022}. 
These works explored the capabilities and limits of algorithms for adversaries that were freely able to choose the input based on previous outputs by the algorithm. 

However, in many realistic settings, adversarial input is limited in its abilities. 
For example, adversarial attacks in machine learning are often permitted to only alter the ``true'' input by a small amount bounded in norm. 
For the $L_0$ norm, this restriction means that the adversary can only add a sparse noise to the true input. 
More generally, it seems reasonable to assume that adversarial input is generated from a source that has bounded computation time or bounded interactions with an honest algorithm. 

\subsection{Our Contributions}
In this paper, we study algorithms robust to adaptive/adversarial input generated from sources with bounded capabilities. 
We first study dynamic algorithms for adaptive inputs from a source that is restricted in sparsity. 
Namely, we consider robust linear algebraic algorithms when the updates to the label can be adversarial but are restricted in sparsity. 
We then study robust algorithms in the standard centralized setting, where an adversary queries an algorithm in an adaptive manner, but the number of interactions between the adversary and the algorithm is bounded. 
We first show that combining novel subroutines for each of these problems in conjunction with a simple but elegant idea of using differential privacy to hide the internal randomness of various subroutines previously used by~\cite{HassidimKMMS20,BeimelKMNSS22,AttiasCSS23} suffices to achieve robust algorithms across these different settings. 

\paragraph{Dynamic algorithms on adaptive input for regression.}
Motivated by the problem of label shift in machine learning, we consider a dynamic version of least-squares regression, where the labels get updated. 
In this model, we are given a fixed design matrix and a  target label that  receives a sequence of updates. 
After each one, the algorithm is asked to output an estimate of the optimal least-squares objective. 
The goal of the algorithm is to maintain the objective value within a multiplicative factor $(1+\eps)$ to the optimal. 

More specifically, the algorithm is given a fixed  design matrix $\bA  \in \R^{n\times d}$ with $n\geq d$ and an initial response vector (i.e., label) $\bb^{(1)}$, which receives updates over time. 
We are interested in estimating the least-squares objective value $F(\bA, \bb) = \underset{\bx\in\mathbb{R}^d}{\min} \|\bA \bx - \bb\|_2^2$ as the target label $\bb$ undergoes updates. 
The updates to $\bb$ are adaptively chosen by an adversary but can only affect at most $K$ entries of $\bb$ per step. Formally, on the $i$-th round:
\begin{enumerate}
    \item The adversary provides an update to $K$ entries of the    $\bb^{(i-1)}$, possibly depending on all previous outputs of the algorithm.
    \item The algorithm updates its data structure and outputs an estimate $\widehat F_i$ of $F_i = F\left(\bA,\bb^{(i)}\right)$.
    \item The adversary observes and records the output $\widehat F_i$.
\end{enumerate}
The goal of the adversary is to create a sequence of labels $\left(\bb^{(i)}\right)_{i=1}^T$ that induces to algorithm to output an inaccurate estimate. 
To deal with adaptivity, a na\"ive idea is to treat each step as an independent least-squares regression problem. 
However, this approach uses a completely new approximation of the objective value for each update, which seems potentially wasteful. 
On the other hand, any randomness that is shared by computations over multiple updates can potentially be leveraged by the adversary to induce an incorrect output. 

Our main result is an algorithm that beats the na\"ive algorithm in this challenging, adaptively adversarial setting. 
We provide a general result with run-time dependence on $n,d,K$,  and the number of nonzero entries in $\bA$, $\nnz(\bA)$. 
\begin{theorem}[Informal; see \autoref{thm:main-dy-reg}]
Let $\kappa(\bA) = \O{1}$ and $\eps \in (0,1)$. 
There exists a dynamic algorithm that given adaptively chosen $K$-sparse updates  to $\bb$ and a fixed design matrix $\bA \in \R^{n\times d}$, outputs a $(1+\eps)$ approximation to the least-squares objective $F(\bA,\bb^{(i)})$ every round with high probability. 
The algorithm uses $\tO{\sqrt{K\nnz(\bA)}/\eps^3}$ amortized time per step of update.  
\end{theorem}
Specifically, the update time  is $d^{1.5}$ when   $K\leq d$ and $n = \O{d}$ and square root of the input sparsity when $K = \O{1}$.  
Notice that this  significantly betters the na\"ive approach of treating each step independently and solving for the least-square objective, which requires  $\O{\nnz(\bA)} + \poly(d)$ time by sketching \cite{woodruff2014sketching}.

We mention that a recent work by \cite{jiang2022dynamic} considers a row-arrival model for dynamic linear regression.  
Our setting is different since we allow arbitrary  updates to the target label, whereas in their setting the design matrix undertakes incremental change. 
We note that their algorithm maintains a solution vector, while we focus on the cost only. 
In particular, approximating the squared error loss is important in applications such as distributed functional monitoring~\cite{CormodeMY11}, where a number of sites are continuously monitored by a central coordinator, who can choose to perform a certain action if the regression cost becomes too high or too low. 
For example, if the cost is too high then perhaps the current set of features needs to be expanded to obtain better prediction, while if the cost is low enough, perhaps the coordinator is satisfied with the current predictor. 
On the other hand, these sites can be sensors, computers, or even entire networks and so certain sites may act in particular ways depending on the actions of the central coordinator. Certain sites may even act maliciously and thus it is important for the algorithm to be adversarially robust.


\paragraph{Robust algorithms in the centralized setting.}
We then consider robust algorithms in the standard \emph{centralized} setting, where an adversary queries an algorithm in an adaptive manner. 
In many key algorithmic applications, randomization is necessary to achieve fast query time and efficient storage. 
This necessitates the need for robust versions of these algorithm which can efficiently employ the power of randomness while also being accurate across multiple possibly correlated inputs. 
Our main parameters of interest are query time and the space used by a robust algorithm compared to their na\"ive, non robust, counterparts.

Formally, we define the model as a two-player game between an algorithm $\honestalg$ over a data set $X$ and an adversary $\adversary$ that makes adversarial queries about $X$ to $\honestalg$. 
At the beginning of the game, $\honestalg$ uses pre-processing time to compute a data structure $\calD$ from $X$ to answer future queries from $\adversary$. 
The game then proceeds in at most $Q$ rounds for some predetermined $Q$, so that in the $t$-th round, where $t\in[Q]$:
\begin{enumerate}
\item
$\adversary$ computes a query $q_t$ on $X$, which depends on all previous responses from $\honestalg$. 
\item
$\honestalg$ uses $\calD$ to output a response $d_t$ to query $q_t$. 
\item
$\adversary$ observes and records the response $d_t$.
\end{enumerate}
The goal of $\adversary$ is to formulate a query $q_t$ for which the algorithm $\honestalg$ produces an incorrect response $d_t$.  
We remark that the algorithm may not have access to $X$, after constructing $\calD$, to respond to the query $q_t$. 
On the other hand, $\adversary$ can use previous outputs to possibly determine the internal randomness of the data structure $\calD$ and make future queries accordingly. 
In this case, the analysis of many randomized algorithms fails because it assumes that the randomness of the algorithm is independent of the input. 
Consequently, it does not seem evident how to handle $Q$ adaptive queries without implementing $Q$ instances of a non-adaptive data structure, i.e., each instance handles a separate query. 
Thus, a natural question to ask is whether a space overhead of $\Omega(Q)$ is necessary. 

\paragraph{Adaptive query framework.}
As a preliminary benchmark, we show that a space overhead of $\Omega(Q)$ is unnecessary by giving a unified framework with only an $\tO{\sqrt{Q}}$ space overhead. 
\begin{theorem}
\thmlab{thm:main:framework}
Given a data structure $\calD$ that answers a query $q$ with probability at least $\frac{3}{4}$ using space $S$ and query time $T$, there exists a data structure that answers $Q$ adaptive queries, with high probability, i.e., $1-\frac{1}{\poly(n,Q)}$, using space $\O{S\sqrt{Q}\log(nQ)}$ and query time $\tO{T\log(nQ)+\log^3(nQ)}$.
\end{theorem}
\thmref{thm:main:framework} invokes the framework of \cite{HassidimKMMS20,BeimelKMNSS22,AttiasCSS23} to the centralized setting, where a number of queries are made only after the data structure is created. 
For completeness, we include the proof in the appendix. 

To concretely instantiate the framework and state an example, we consider the adaptive distance estimation problem defined as follows. 
In the adaptive distance estimation problem, there exists a set $X=\{\bx^{(1)},\ldots,\bx^{(n)}\}$ of $n$ points in $\mathbb{R}^d$. 
We are also given an accuracy parameter $\eps>0$. 
A query is of the form $\bq$, and the algorithm must output a $(1+\eps)$-approximation to $\|\bx^{(i)}-\bq\|_p$ for all $i$.
The trivial solution of storing all $n$ points and computing all $n$ distances to a query point uses space and query time $\O{nd}$. 
\cite{CherapanamjeriN20} improved the query time to $\tO{\frac{n+d}{\eps^2}}$ at the cost of using $\tO{\frac{(n+d)d}{\eps^2}}$ space and $\tO{\frac{nd^2}{\eps^2}}$ pre-processing time, while permitting an arbitrary number of queries. 
By comparison, our data structure handles $Q$ queries of approximate distances from a \emph{specified point in $X$}, using query time $\tO{\frac{n+d}{\eps^2}}$, pre-processing time $\tO{\frac{nd\sqrt{Q}}{\eps^2}}$, and space $\tO{\frac{(n+d)\sqrt{Q}}{\eps^2}}$. 
Thus, in the regime where $d\gg n\sqrt{Q}$, the data structure already improves on the work of \cite{CherapanamjeriN20}. 

However, a noticeable weakness of the construction is that the $Q$ queries return only the approximate distance between a query point and a single point in $X$, whereas \cite{CherapanamjeriN20} outputs approximate distances to all points in $X$. 
Moreover, \cite{CherapanamjeriN22} subsequently improve the pre-processing time to $\tO{\frac{nd}{\eps^2}}$. 
Thus we open up the framework to (1) show that it can be further improved to handle the case where we return the approximate distances of all points in $X$ from $Q$ adaptive query points and (2) achieve pre-processing time $\tO{\frac{nd}{\eps^2}}$. 

\begin{theorem}
\thmlab{ade_srht}
There is a data structure which, when instantiated with dataset $X = \{x_i\}_{i \in [n]} \subset \R^d$ and query bound $Q \leq d$, answers any sequence of $Q$ adaptively chosen distance estimation queries correctly with probability at least $0.99$. Furthermore, the space complexity of the data structure is $\wt{O}(\eps^{-2} \cdot n \sqrt{Q})$ and the setup and query times are $\wt{O}(\eps^{-2} \cdot nd)$ and $\wt{O} (\eps^{-2} \cdot (n + d))$, respectively. 
\end{theorem}

Another application of the framework is the adaptive kernel density estimation problem, where there exists a set $X=\{\bx^{(1)},\ldots,\bx^{(n)}\}$ of $n$ points in $\mathbb{R}^d$ and the goal is to output a $(1+\eps)$-approximation to the quantity $\frac{1}{n}\sum_{i\in[n]}k(\bx^{(i)},\bq)$, for an accuracy parameter $\eps>0$, a query $\bq$, and a kernel function $k$, under the promise that the output is at least some threshold $\tau>0$. 
\cite{BackursIW19} give an algorithm for kernel density estimation that uses $\O{\frac{1}{\tau\eps^2}}$ space and $\O{\frac{d}{\sqrt{\tau}\eps^2}}$ query time, improving over the standard algorithm that samples $\O{\frac{1}{\tau\eps^2}}$ points and then uses $\O{\frac{d}{\tau\eps^2}}$ query time to output the empirical kernel density. 
However, the analysis for both of these algorithms fails for the adaptive setting, where there can be dependencies between the query and the data structure. 
By using the data structure of \cite{BackursIW19} as a subroutine, the framework immediately implies an algorithm for adaptive kernel density estimation that uses $\tO{\frac{\sqrt{Q}}{\tau\eps^2}}$ space and $\O{\frac{d\log Q}{\sqrt{\tau}\eps^2}}$ query time to answer each of $Q$ adaptive queries. 
In this case, we are again able to go beyond the framework and give a data structure that handles an unlimited number of adaptive kernel density queries:

\begin{restatable}{theorem}{thmakdeinfty}
\thmlab{thm:akde:infty}
Suppose the kernel function $k$ is $L$-Lipschitz in the second variable for some $L>0$, i.e., $|k(\bx,\by)-k(\bx,\bz)|\le L\|\by-\bz\|_2$ for all $\bx,\by,\bz\in\mathbb{R}^d$. 
Moreover, suppose that for all $\|\bx-\by\|_2\le\rho$, we have $k(\bx,\by)\le\frac{\tau}{3}$.  
Then an algorithm that produces a kernel density estimation data structure $D$ that is $L$-Lipschitz over a set $X$ of points with diameter at most $\Delta$ and outputs a $(1+\eps)$-approximation to KDE queries with value at least $\tau$ with probability at least $1-\delta$ using space $S(n,\eps,\tau,\log\delta)$ and query time $T(n,\eps,\tau,\log\delta)$, then there exists a KDE data structure that with probability at least $0.99$, outputs a $(1+\eps)$-approximation to any number of KDE queries with value at least $\tau$ using space $S\left(n,\O{\eps},\O{\tau},\O{d\log\frac{(\Delta+\rho) L}{\eps\tau}}\right)$ and query time $T\left(n,\O{\eps},\O{\tau},\O{d\log\frac{(\Delta+\rho) L}{\eps\tau}}\right)$.
\end{restatable}

Additionally, we show that the framework guarantees adversarial robustness for a number of other important problems such as nearest neighbor search, range queries, point queries, matrix-vector norm queries, and linear regression. 
Finally, we supplement our theoretical results with a number of empirical evaluations. 

\subsection{Our Techniques}

\paragraph{Dynamic regression on adaptive inputs.} 
Our dynamic algorithm for dynamic maintenance of the least-squares objective exploits two main ideas. First, standard results in sketching and sampling show that it suffices to solve for the sketched objective of  $\min_{\bx \in \R^{d}}\|\bS\bA \bx -\bS \bb\|_2^2$, where $\bS$ is an $\ell_2$ subspace embedding for $\bA$. 
Here, we   exploit  several techniques from the  numerical linear algebra  and in particular use leverage score sampling to obtain a   subspace embedding   $\bS$ of $\bA$. By standard results in sketching, a $(1+\eps)$ optimal solution is given by $\bx^*= (\bS\bA)^\dagger \bS\bb$.
Moreover, since the goal is to output the objective value instead of the solution vector, we may take a Johnson-Lindenstrauss (JL) sketch to further  reduce dimensionality and  run-time. This allows us to focus on   $\|\bG\bA \bx^* -\bG \bb\|_2^2$,
where $\bG \in \R^{\O{\log d} \times n}$ is a JL sketch. 

As a result, our algorithm dynamically maintains a solution $\bG\bA(\bS\bA)^\dagger \bb$ in this sketched space. 
To achieve that, we first explicitly solve $\bG\bA(\bS\bA)^\dagger$ in pre-processing. 
Since $\bG\bA$ has few rows, this reduces to a small number of linear solves and can be computed fast via conjugate gradient-type methods. 
To handle the updates, we leverage their sparsity to efficiently maintain the solution  and show that each round takes roughly $\O{K}$ time. Amortizing the pre-processing with the update  costs over all iterations yields our desired run-time.

Finally, we apply techniques from differential privacy~\cite{HassidimKMMS20,KaplanMNS21,BeimelKMNSS22,AttiasCSS23} to ensure adversarial robustness, by aggregating independent copies of the algorithm via a private median mechanism.
Intuitively, the private mechanism hides the internal randomness of the algorithm and therefore prevents the adversary from otherwise choosing a ``bad'' input based on knowledge of internal parameters.

\paragraph{Adaptive query framework.}
Our framework maintains $\tO{\sqrt{Q}}$ instances of the non-adaptive data structure and crucially uses differential privacy (DP) to protect the internal randomness of the data structures. 
In addition to our previous results for dynamic regression and $k$-cut queries on adaptive input, the technique of using DP to hide randomness has recently been used in the streaming model~\cite{HassidimKMMS20,KaplanMNS21,AttiasCSS23} and the dynamic model~\cite{BeimelKMNSS22}. 
These works elegantly use the advanced composition property of DP to bound the number of simultaneous algorithms that must be used in terms of the number of times the output changes ``significantly'' over the course of the stream. 
In the streaming model, the robust algorithms proceed by instantiating many ``hidden" copies of a standard randomized algorithm. 
As the stream arrives, the algorithms are updated and an answer, aggregated using DP, is reported. 
Crucially, many of these results exploit the fact that the output answer is monotonic in the stream so that there is a known upper bound on the final output. 
Thus, the reported answers can only increase by a multiplicative factor at most a logarithmic number of times, which is used to bound the initial number of algorithms which are initialized. 
In the centralized setting, this can be imagined as setting the parameter $Q$. 
The main parameter of interest in the streaming literature is the space used by the streaming algorithms, whereas we are concerned with both space usage and query times. 
Furthermore, stream elements are only accessed one at a time and cannot be processed together unless memory is used. 
In our case, the dataset is given to us upfront and we can pre-process it to construct a data structure towards solving a centralized problem. 

The work \cite{BeimelKMNSS22} shares many of these ideas: the authors are concerned with dynamic graph algorithms where an adversary can update the graph in an adaptive fashion. 
Similar tools such as multiple randomized initialization and aggregated responses using DP are utilized. 
The main difference is their parameters of interest: the goal of \cite{BeimelKMNSS22} is to have a fast \emph{amortized} update time across many queries. 
This necessitates the need to ``throw away" existing algorithms and start with fresh randomness at intermittent points. 
In contrast, we study a centralized setting where the underlying dataset is not updated but we wish to answer $Q$ adaptive queries on the dataset.

By the same reasoning, advanced composition can be used to show the sufficiency of maintaining $\tO{\sqrt{Q}}$ data structures to answer $Q$ adaptive queries in the centralized setting, which gives a rich set of applications. 
Moreover, to improve the query time of the framework, the privacy amplification of sampling can be further invoked to show that it suffices to output the private median of a small subset, i.e., a subset of size $\O{\log Q}$, of these $\tO{\sqrt{Q}}$ data structures. 
Thus the framework only incurs a logarithmic overhead in query time and an $\tO{\sqrt{Q}}$ overhead in space. 
Surprisingly, the simple framework gives diverse applications for adaptive algorithms on a number of important problems, including estimating matrix-vector norms, adaptive linear regression, adaptive range query search, adaptive nearest neighbor search, and adaptive kernel density estimation, to name a few. 
For completeness, these applications are discussed in depth in the appendix in \secref{sec:framework}. 

We emphasize that for several applications of our framework such as adaptive distance estimation or adaptive kernel density estimation, we additionally use specific sophisticated techniques for these problems to further improve our algorithmic guarantees. 
As a simple example, for adaptive kernel density estimation, we provide a data structure robust to an arbitrary number of adaptive queries, which cannot be handled by the techniques of \cite{BeimelKMNSS22}. 

\paragraph{Adaptive distance estimation.} 
To achieve better pre-processing time for adaptive distance estimation, our main technique is to sample groups of rows from a Hadamard transform and argue that an interaction with a separate group should be considered in separate privacy budgets, effectively arguing that outputting $n$ approximate distances to a single adaptive query only uses one unit of privacy budget. 
By contrast, our black-box framework charges one unit of privacy budget per approximate distance, so that outputting $n$ approximate distances would use $n$ units of privacy budget. 

\paragraph{Adaptive kernel density estimation.} 
\thmref{thm:akde:infty} is based on showing that with constant probability, our data structure is accurate on all possible queries in $\mathbb{R}^d$. 
In particular, we first show that our data structure is accurate on a sufficiently fine net of points through a standard union bound argument, which incurs the $d$ overhead compared to the space required to handle a single query. 
We then show that if the algorithm and the kernel function are both Lipschitz, which is true for sampling-based algorithms and a number of standard kernel functions, then accuracy on the net implies accuracy on all possible points in $\mathbb{R}^d$. 

\section{Preliminaries}
\paragraph{Notations.} 
In this paper, we use $[n]$ for a positive integer $n>0$ to denote the set $\{1,\ldots,n\}$. 
We use $\poly(n)$ to denote a fixed polynomial in $n$. 
We say an event occurs with high probability if it occurs with probability $1-\frac{1}{\poly(n)}$.
For real numbers $a,b$ and positive $\eps$, we say $a = (1\pm \eps) b$ if $(1-\eps)b \leq a \leq (1+\eps)b$.    
Let $\be_i\in\R^{\rows}$ be the $i$'th standard basis vector.
Let $\bX^+$ denote the Moore-Penrose pseudo-inverse of matrix $\bX$. Let $\norm{\bX}$ denote the  operator norm of $\bX$.
Let $\kappa(\bX)=\norm{\bX^+}\norm{\bX}$ denote the condition number of $\bX$.

\subsection{Differential Privacy}
Much of our technical results leverage tools from DP. 
We recall its definition and  several key statements.
\begin{definition}[Differential privacy, \cite{DworkMNS06}]
Given $\eps>0$ and $\delta\in(0,1)$, a randomized algorithm $\calA:\calX^*\to\calY$ is $(\eps,\delta)$-differentially private if, for every neighboring datasets $S$ and $S'$ and for all $E\subseteq\calY$,
\[\PPr{\calA(S)\in E}\le e^{\eps}\cdot\PPr{\calA(S')\in E}+\delta.\]
\end{definition}

\begin{theorem}[Amplification via sampling, e.g.,~\cite{BunNSV15}]
\thmlab{thm:dp:sampling}
Let $\calA$ be an $(\eps,\delta)$-differentially private algorithm for $\eps\le 1$, $\delta\in(0,1)$. 
Given a database $S$ of size $n$, let $\calA'$ be the algorithm that constructs a database $T\subset S$ by subsampling (with replacement) $s\le\frac{n}{2}$ rows of $S$ and outputs $\calA(T)$. 
Then $\calA'$ is $(\eps',\delta')$-differentially private for
\[\eps'=\frac{6\eps k}{n},\qquad\delta'=\exp(6\eps k/n)\,\frac{4k\delta}{n}.\]
\end{theorem}

\begin{theorem}[Private median, e.g.,~\cite{HassidimKMMS20}]
\thmlab{thm:dp:median}
Given a database $\calD\in X^*$, there exists an $(\eps,0)$-differentially private algorithm $\privmed$ that outputs an element $x\in X$ such that with probability at least $1-\delta$, there are at least $\frac{|S|}{2}-k$ elements in $S$ that are at least $x$, and at least $\frac{|S|}{2}-k$ elements in $S$ in $S$ that are at most $x$, for $k=\O{\frac{1}{\eps}\log\frac{|X|}{\delta}}$. 
\end{theorem}

\begin{theorem}[Advanced composition, e.g.,~\cite{DworkRV10}]
\thmlab{thm:adaptive:queries}
Let $\eps,\delta'\in(0,1]$ and let $\delta\in[0,1]$. 
Any mechanism that permits $k$ adaptive interactions with mechanisms that preserve $(\eps,\delta)$-differential privacy guarantees $(\eps',k\delta+\delta')$-differential privacy, where $\eps'=\sqrt{2k\ln\frac{1}{\delta'}}\cdot\eps+2k\eps^2$. 
\end{theorem}

\begin{theorem}[Generalization of DP, e.g.,~\cite{DworkFHPRR15,BassilyNSSSU21}]
\thmlab{thm:generalization}
Let $\eps\in(0,1/3)$, $\delta\in(0,\eps/4)$, and $n\ge\frac{1}{\eps^2}\log\frac{2\eps}{\delta}$. 
Suppose $\calA:X^n\to 2^X$ is an $(\eps,\delta)$-differentially private algorithm that curates a database of size $n$ and produces a function $h:X\to\{0,1\}$. 
Suppose $\calD$ is a distribution over $X$ and $S$ is a set of $n$ elements drawn independently and identically distributed from $\calD$. 
Then
\[\PPPr{S\sim\calD,h\gets\calA(S)}{\left|\frac{1}{|S|}\sum_{x\in S}h(x)-\EEx{x\sim\calD}{h(x)}\right|\ge10\eps}<\frac{\delta}{\eps}.\]
\end{theorem}

\subsection{Numerical Linear Algebra}
Our results on dynamic regression relies upon some   tools in numerical linear algebra. We first recall the  dimensionality reduction techniques.
\begin{theorem}[Johnson-Lindenstrauss transformation, $\eps$-JL]
\thmlab{thm:jl}
Given $\eps>0$, there exists a family of random maps $\Pi_{m,d}\in\mathbb{R}^{m\times d}$ with $m=\O{\frac{1}{\eps^2}}$ such that for any $\bx\in\mathbb{R}^d$, we have
\[\PPPr{\Pi\sim\Pi_{m,d}}{(1-\eps)\|\bx\|_2\le\|\Pi\bx\|_2\le(1+\eps)\|\bx\|_2}\ge\frac{3}{4}.\]
Moreover, $\Pi\bx$ takes $\O{\frac{d}{\eps^2}}$ time to compute. 
\end{theorem}

\begin{theorem}[Fast JL]
\thmlab{thm:fast:jl}
Given $\eps>0$, there exists a family of random maps $\Pi_{m,d}\in\mathbb{R}^{m\times d}$ with $m=\O{\frac{\log d}{\eps^2}}$ such that for any $\bx\in\mathbb{R}^d$, we have
\[\PPPr{\Pi\sim\Pi_{m,d}}{(1-\eps)\|\bx\|_2\le\|\Pi\bx\|_2\le(1+\eps)\|\bx\|_2}\ge\frac{3}{4}.\]
Moreover, $\Pi\bx$ takes $\O{\frac{\log d}{\eps^2}+d\log d}$ time to compute. 
\end{theorem}
A row sampling matrix $\bS$ has rows that are multiples of natural basis vectors, so that $\bS\bA$ is
a (weighted) sample of the rows of $\bA$. A column sampling matrix is defined similarly.
The size of a row/column sampling matrix is defined as the number of rows/columns it samples. 
The leverage score of the $i$th row $\ba_i^\top$ of $\bA$ is 
\[
\tau_i(\mathbf{A}) \stackrel{\text { def }}{=} \mathbf{a}_i^{\top}\left(\mathbf{A}^{\top} \mathbf{A}\right)^{+} \mathbf{a}_i.\]
For a   survey on leverage score and applications, we refer the reader to \cite{mahoney2011randomized}.
\begin{definition}[Leverage score sampling]\label{lem:lv-compute}
    Let $\bu$ be a vector of leverage score overestimates, i.e., $\tau_i(\bA) \leq \bu_i $. 
    Let  $\alpha$ be a
sampling rate parameter and   $c$ be a fixed positive constant. For each row, we define a sampling probability $p_i=\min \left\{1, \alpha \cdot u_i c \log d\right\}$.  The leverage score sampling matrix is a row sampling matrix $\mathbf{S}$ with independently chosen entries such that $\mathbf{S}_{i i}=\frac{1}{\sqrt{p_i}}$ with probability $p_i$ and $0$ otherwise.
\end{definition}
\begin{definition}[Subspace embedding]\label{def:sub-emb}
A  $(1 \pm \varepsilon) $ $\ell_2$ subspace embedding for the column space of an $n \times d$ matrix $\mathbf{A}$ is a matrix $\mathbf{S}$ for which for all $\mathbf{x} \in \mathbb{R}^d$
$$
\|\mathbf{S A x}\|_2^2=(1 \pm \varepsilon)\|\mathbf{A} \mathbf{x}\|_2^2.
$$
\end{definition}

\begin{theorem}[Leverage sampling implies subspace embedding, Theorem 17 of \cite{woodruff2014sketching}]\label{lem:lev-sub-emb}
Let $\alpha = \eps^{-2}$ and $c$ be a  sufficiently large constant. With high probability, the leverage score sampling matrix  is a $(1\pm\eps)$ $\ell_2$ subspace embedding. Furthermore, it has size $\O{d\log d/ \eps^2}$.
\end{theorem}




The approximate leverage scores can be computed in input-sparsity time. Afterwards, repeated sampling from the leverage score distribution can be done efficiently using the binary tree data structure in quantum-inspired numerical linear algebra.
\begin{lemma}[Leverage score computation and sampling data structure;  see \cite{woodruff2014sketching,chepurko2022quantum}]\label{lem:dynsamp}
Let $\bA \in \R^{n\times d}$. 
There exists an algorithm that given $\bA$ outputs a vector of row leverage score overestimates with high probability and in run-time $\tO{\nnz(\bA) + \poly(d)}$.

Furthermore, there exists a sampling data structure $\mathcal D_{LS}$ that stores the row  leverage scores of $\bA$ such that given a positive integer $m\leq n$, returns a leverage score sample of $\bA$  of size $m$ in $\O{m\log (mn)}$ time.
In total, the pre-processing takes $\O{\nnz(\bA) + \poly(d)}$ time. 
\end{lemma}

\section{Dynamic Regression}
\label{sec:dynamic-regression}
In this section, we consider the dynamic problem of maintaining the cost of the least-squares regression, where the labels receive adaptively chosen updates.

We first introduce the basic setting of the problem in \secref{sec:setting-lsr}. 
In  \secref{sec:dynamic-ob-reg}, we design a key subroutine under   non-adaptive  updates. 
The data structure enjoys a nearly linear update time. 
This allows us to aggregate multiple copies of the procedure and thereby efficiently ensure adversarial robustness against an adaptive adversary. 
The argument is via an application of differential privacy and detailed subsequently in \secref{sec:dynamic-ada-reg}.

\subsection{Basic Setting}\seclab{sec:setting-lsr}
Let $\bA \in \R^{n\times d}$ be the design matrix and $\bb \in \R^n$ be the target label. 
A classic problem in numerical linear algebra and optimization is to solve the $\ell_2$ least-squares regression objective   
\begin{equation}
   F(\bA, \bb) = \min_{\bx \in \R^{d}}\left\|\bA \bx - \bb\right\|_2^2 = \left\|\bA \bA^\dagger \bb - \bb\right\|_2^2.
\end{equation}
We consider a dynamic version of the problem, where the design matrix $\bA$ remains unchanged. However, at each step (at most) $K$ entries of  $\bb$ undergo  an update.
Moreover, we assume that the updates are chosen adaptively by an adversary in the following manner.
\begin{itemize}
    \item The algorithm starts by receiving the input $\bA \in \R^{n\times d}$ and   $\bb^{(1)} \in \R^n$.
    \item In the $i$-th step,   the algorithm  outputs an estimate $\widehat{F}_i$ of the cost $F(\bA,\bb^{(i)})$, where $\bb^{(i)}$ is the target label   corresponding to the   step. 
    \item The adversary observes $\widehat{F}_i$ and  updates  at most $K$ labels   to form $\bb^{(i)} $.
\end{itemize}
Let $\bb^{(1)},\bb^{(2)},\ldots, \bb^{(T)} \in \R^{n}$  be the resulting sequence of labels over $T$ steps.  
The goal of the algorithm is to output a $(1+\eps)$ approximation to the optimal cost at every step, while minimizing  the     update time.

%


\subsection{Dynamic Algorithm for Oblivious Inputs}
\seclab{sec:dynamic-ob-reg}
In this section, we provide a key subroutine that maintains a data structure under oblivious updates. 
On a high-level, the  data structure aims to enable a \textit{sketch-and-solve} strategy dynamically. The main ideas are two fold: (1) apply  randomized sketching   to reduce dimensionality and therefore the run-time, and (2)  exploit the sparsity of the updates to argue that the regression costs can be maintained efficiently. 

Before delving into the technical details, we give an overview of the algorithm.

\paragraph{Overview of the algorithm.} 
We start by assuming that the algorithm has access to $\mathcal{D}_{LS}$ (via \autoref{lem:dynsamp}), the   row leverage score sampling data structure for $\bA$. 
In preprocessing, the algorithm samples a leverage score sketching matrix $\bS \in \R^{k \times n}$ from  $\mathcal{D}_{LS}$, where $k= \mathcal O({d\log d/\eps^2})$.  This  provides a $(1+\eps)$ $\ell_2$ subspace embedding for $\bA$.  Standard results in sketching imply that it suffices to solve for the sketched objective of  $\min_{\bx \in \R^{d}}\|\bS\bA \bx -\bS \bb\|_2^2$ \cite{sarlos2006improved,ClarksonW13,clarkson2017low, woodruff2014sketching}. 
Let $\widehat \bA = \bS\bA$. 
Then a $(1+\eps)$ optimal solution is thus given by $\widehat{\bA}^\dagger \bb$. 
Moreover, our goal is to maintain the regression cost, rather than this solution vector. 
Hence, we  can apply Johnson–Lindenstrauss lemma and focus on 
\begin{equation}\label{eqn:sketched-lsq}
    \min_{\bx \in \R^{d}}\left\| \bS\bA \bx -\bS \bb\right\|_2^2 \approx \left\|\bG \bA (\bS\bA)^\dagger\bS \bb -\bG  \bb\right\|_2^2,
\end{equation}
where $\bG \in \R^{\mathcal O(\log n/\eps^2) \times n}$ is a JL sketch.

Next, we describe how to track the cost value dynamically.
We stress that the sketching matrices $\bS$ and $\bG$ are sampled upfront in the preprocessing  stage and remain fixed afterwards. 
The algorithm  stores  $\bG$ and  $\bM = \bG  \bA (\bS\bA)^\dagger$, both computed in preprocessing.
Meanwhile, it maintains $\bG\bb^{(i)},\bS\bb^{(i)}$, initialized at $i=1$. 
In the first step, given the initial target label $\bb^{(1)}$, the algorithm computes $\bS\bb^{(i)}$, $\bM \left(\bS \bb^{(1)}\right)$ and $\bG  \bb^{(1)}$. Then it outputs $\widehat F_1 = \left\|\bM\bS \bb^{(1)}- \bG  \bb^{(1)}\right\|_2^2$ as an estimate of the regression cost. 

Let's consider the $i$-th step, where the label is updated to $\bb^{(i)}$. 
First, we read the $K$ labels that get changed and update $\bS\bb^{(i-1)}$ to $\bS\bb^{(i)}$ accordingly. This can be done in $\mathcal{O}(K)$ time. Finally, we simply compute $\bM (\bS \bb^{(i)})$ and $\bG  \bb^{(i)}$ and output $\widehat F_i = \left\|\bM\bS \bb^{(i)}- \bG  \bb^{(i)}\right\|_2^2$. We store $\bG\bb^{(i)}$ for the next iteration.

We now describe the algorithm formally, followed by an analysis of its run-time and accuracy.

\paragraph{Formal description of the algorithm.}
We assume $\mathcal{D}_{LS}$ for $\bA$ is given.
The data structure is initialized by drawing the sketching matrices $\bG$ and $\bS$. 
We also compute 
    $\bM = \bG \bS\bA (\bS\bA)^\dagger$ in preprocessing.
   This matrix is   stored explicitly throughout.

\begin{algorithm}[!htb]
\caption{Initialize the data structure, i.e., preprocessing}
\alglab{alg:init-reg}
\begin{algorithmic}[1]
\Require{Design matrix $\bA \in \R^{n\times d}$, initial label $\bb^{(1)}\in \R^{n}$, $\mathcal{D}_{LS}$, $\eps \in (0,1)$} 
\Ensure{Preprocessing matrix $\bM$}
\State{Let $k =\Theta \left(d \log d  /\eps^{2} \right)$}
\State{Sample a $(1+\eps/2)$  $\ell_2$ leverage score row sampling matrix $\bS\in \R^{k\times n}$ for $\bA$ from $\mathcal{D}_{LS}$.} 
\State{Sample a JL sketch matrix $\bG  \in \R^{C\eps^{-2}\log n\times n}$, for a sufficiently large $C$, by \thmref{thm:jl}.}
\State{Compute  and store  $\bM = \bG  \bA (\bS\bA)^\dagger$.}
\end{algorithmic}
\end{algorithm}

At each   step,   the algorithm computes $\bS\bb^{(i)}$ by reading all $K$ entries of $\bb^{i-1}$ that are updated in the step. After that, compute 
$\bM(\bS \bb^{(i)} )$ and $\bG\bb^{(i)}$ and output $\left\|\bM \bb^{(i)}- \bG\bb^{(i)}\right\|_2^2$. The algorithm is formally given by \algref{alg:update-reg}. 
 
\begin{algorithm}[!htb]\caption{Update  data structure and maintain regression cost}
\alglab{alg:update-reg}
\begin{algorithmic}[1]
\Require{Matrices $\bG \in \R^{C\eps^{-2}\log n\times n}, \bS\in\R^{k\times n}, \bM \in \R^{\tO{1/\eps^2} \times k}$ and   the   label $\bb^{(i)}$}
\Ensure{Estimate of the regression cost $F\left(\bA ,\bb^{(i)} \right)$}
 
\State{Compute $\bS\bb^{(i)}$ by reading all $K$ entries of $\bb^{(i-1)}$ that are updated.}
\State{Compute  
$\bM \left(\bS \bb^{(i)}\right)$ and $\bG  \bb^{(i)}$.}
\Comment{Store $\bM\bS  \bb^{(i)},\bS  \bb^{(i)}$, $\bG  \bb^{(i)}$ for the next round.}
\State{Output $\widehat{F}_i = \left\|\bM \bS\bb^{(i)}- \bG \bb^{(i)}\right\|_2^2$.}
\end{algorithmic}
\end{algorithm}

\paragraph{Analysis of the algorithm.} We now analyze the run-time of the algorithm. First, consider the preprocessing stage performed by \algref{alg:init-reg}.

\begin{lemma}[Preprocessing time]
\label{lem:preproc-time}
Assuming access to the leverage score sampling data structure $\mathcal{D}_{LS}$, the preprocessing time of   \algref{alg:init-reg} is 
\begin{equation}
\O{\sqrt{\kappa (\bA)} \nnz(\bA)\log\frac{1}{\eps} +\frac{\nnz(\bA)}{\eps^2}\log n+\frac{d}{\eps^2}\log n}.
\end{equation}
\end{lemma}
\begin{proof} 
By \autoref{lem:dynsamp}, the guarantee of the  sampling data structure $\mathcal{D}_{LS}$, it takes $\mathcal{O}(k \log(nd))$ time  to obtain a leverage score sample $\bS$ of size $k$. Drawing the JL sketch is straightforward, and standard constructions such as i.i.d.\ Gaussian entries require $\mathcal{O}(k\log n /\eps^2)$ times to form $\bG$.

Finally, we need to compute $\bM$. 
Computing $\bG\bA$ requires $\O{\frac{\nnz(\bA)}{\eps^2}\log n}$ time by sparse matrix multiplication. 
Moreover, since $\bG\bA$ is a matrix of $\O{\frac{\log n}{\eps^2}}$ rows, then computing $(\bG\bA) (\bS\bA)^\dagger$ reduces to $\O{\frac{\log n}{\eps^2}}$ number of linear system solves with respect to $\bS\bA \in \R^{k\times d}$. 
By conjugate gradient type methods, since $\kappa(\bS\bA)  = (1\pm\eps) \kappa(\bA)$, each solve can be achieved to high accuracy in $\mathcal O\left({\sqrt{\kappa(\bA)} \log (1/\eps)}\right)$ number of matrix-vector products with respect to $\bA$ \cite{golub2013matrix}.
In total, this gives a run-time of $\O{\sqrt{\kappa (\bA)} \nnz(\bA)\log(1/\eps)} $.
\end{proof}

\begin{lemma}[Update time]
\label{lem:update-time}
The update time of   \algref{alg:update-reg} is $\O{\frac{K}{\eps^2}\log n}$ per step.
\end{lemma}
\begin{proof}
First, the algorithm reads the $K$ entries that are updated and compute the  $\bS\bb^{(i)}$ from $\bS\bb^{(i-1)}$. This step takes $\mathcal{O}(K)$ time, since we just need to update the entries that lie in the support of the row sampling matrix $\bS$. 
Similarly, in step 2 of  \algref{alg:update-reg}  we can update $\bG\bb^{(i-1)}$ to $\bG\bb^{(i)}$ in $\mathcal{O}(K\log n /\eps^2)$ time.
Since $\bS$ is a row sampling matrix and $\bb^{(i)}$ only has $K$ entries updated, then $\bS\bb^{(i)}$ has at most $K$ entries updated as well.
It follows that given  $\bM \left(\bS \bb^{(i-1)}\right)$ from the prior round, $\bM \left(\bS \bb^{(i)}\right)$ can be updated in $\O{\frac{K}{\eps^2}\log n}$ time.
\end{proof}

\begin{lemma}[Accuracy]
    Given a stream of $T= \mathcal{O}(d^2)$ non-adaptive updates and error parameter $\eps \in (0,1/4)$,
    \algref{alg:update-reg} outputs an estimate $\widehat F_i$ of  the regression cost $F(\bA, \bb^{(i)})$ such that $\widehat F_i = (1\pm \eps) F(\bA, \bb^{(i)})$ for all $i$ with high probability. 
\end{lemma}
\begin{proof}
First, we apply  the subspace embedding property of $\bS$. 
 This implies that with high probability,
    \begin{equation*}
        \min_{\bx}\left\| \bS \bA \bx - \bS\bb^{(i)}\right\|^2_2  = (1\pm \eps/2)   \min_{\bx}\left\|  \bA \bx - \bb^{(i)}\right\|^2_2.
    \end{equation*}
      Apply the JL lemma (\thmref{thm:jl}), where we consider the collection of $\O{d^2}$ $(1+\eps)$ optimal predictions $\{\by_i^*\}_{i=1}^T$ with $\by_i^* = \bA (\bS\bA)^{\dagger} \bb^{(i)}$.  
    Via union bound, we have  that with high probability for all $i \in [T]$
    \begin{equation*}
        \left\|\bG\by_i^* - \bG \bb^{(i)}\right\|_2^2 = (1\pm \eps /2)  \left\|\by_i^* -   \bb^{(i)}\right\|_2^2.
    \end{equation*}
    Our algorithm precisely solves for $\by_i^*$ each iteration. Combining the   two equations above finishes the proof.
\end{proof}

\subsection{Dynamic Algorithm with Adversarial Robustness}
\seclab{sec:dynamic-ada-reg}
To put everything together and ensure adversarial robustness, we use a standard approach of \cite{HassidimKMMS20,BeimelKMNSS22,AttiasCSS23}. 
Our full algorithm maintains $\Gamma = \O{\sqrt{T} \log(nT)}$ independent copies of the key subroutine for  $T=\O{\frac{\nnz(\bA)}{\eps^2 K}}$. 
Then at each step, we output the private median of the outputs of these copies. 
Advanced composition of DP ensures robustness up to $T$ rounds. 
Afterwards, the algorithm reboots by rebuilding the copies, using fresh randomness independently for sampling and computing the sketching matrices.

\begin{algorithm}[!htb]
\caption{Preprocessing  step for \algref{alg:median-reg}}
\alglab{alg:median-reg-prep}
\begin{algorithmic}[1]
\Require{A design matrix $\bA \in \R^{n\times d}$, an approximation  factor $\eps \in (0,1)$.}
\Ensure{The leverage score sampling data structure $\mathcal{D}_{LS}$ for $\bA$.}  
\State{Compute the approximate row leverage scores of $\bA$.} 
\Comment{\autoref{lem:dynsamp}}
\State{Build and output the data structure $\mathcal{D}_{LS}$}
\end{algorithmic}
\end{algorithm}

\begin{algorithm}[!htb]
\caption{Dynamic algorithm for maintaining regression cost under adaptive updates}
\alglab{alg:median-reg}
\begin{algorithmic}[1]
\Require{A sequence of target labels $\left\{\bb^{(i)}\right\}_{i=1}^m$ and a fixed design matrix $\bA \in \R^{n\times d}$, an approximation  factor $\eps \in (0,1)$, the leverage score sampling data structure $\mathcal{D}_{LS}$ for $\bA$.}
\Ensure{Estimates of the regression cost $F(\bA ,\bb^{(i)} )$ under adaptively chosen updates to $\bb$.}  
\For{every epoch of $T=\O{\frac{\nnz(\bA)}{\eps^2 K}}$ updates}
\State{Initialize $\Gamma = \mathcal O\left(\sqrt{T} \log {(nT)}\right)$ independent instances of the data structure in   \secref{sec:dynamic-ob-reg} via \algref{alg:init-reg}.}
\State{Run \textsf{PrivMed} on the $\Gamma$ instances with privacy parameter $\eps' = \mathcal{O} \left(\frac{1}{\sqrt{T}\log (nT)}\right)$ with failure probability $\delta = \frac{1}{\text{poly}(m,T)}$.}
\State{For each query, return the output of  \textsf{PrivMed}.}
\EndFor
\end{algorithmic}
\end{algorithm}

\begin{restatable}{theorem}{thmmaindyreg}[Main theorem; dynamic maintenance of regression cost]
\label{thm:main-dy-reg}
Let $\eps \in (0,1/4)$ be an   error parameter and $\bb^{(1)}$ be the initial target label. 
Given  $\eps, \bA, \bb^{(1)}$, a stream of $T$  adaptively chosen, $K$-sparse updates to the  label, \algref{alg:median-reg} outputs an estimate $\widehat F_i $ such that $\widehat F_i = (1\pm\eps) F(\bA , \bb^{(i)})$ for all $i$ with high probability. 
 
Furthermore, the algorithm requires a preprocessing step in time $\tO{\nnz(\bA) + \poly(d)}$. The amortized update time of the algorithm is 
\[ \tO{  \sqrt{K\nnz(\bA) }  \left(\sqrt{\kappa(\bA)} + \eps^{-3}\right)  }\] 
per round.
\end{restatable}
\begin{proof}
We focus on any fixed epoch of $T$ iterations.  Let $\{\calA_i\}_{i=1}^\Gamma$ be the collection of $\Gamma$ data structures  maintained by the \algref{alg:median-reg} and $\calT_i$ be the transcript between     \algref{alg:median-reg}  and the adversary at round $i$, consisting of the algorithm's output and the update requested by the adversary.

To handle a sequence of $T$ adaptive queries, consider the transcript $\calT(R)=\{\calT_1,\ldots,\calT_T\}$, where $R$ denotes the internal randomness of \algref{alg:median-reg}.
Note that for a fixed iteration, $\calT_i$ is $\left(\O{\frac{1}{\sqrt{T}\log(nT)}},0\right)$-differentially private with respect to the algorithms $\calA_1,\ldots,\calA_\Gamma$, since the private median algorithm $\privmed$ is $\left(\O{\frac{1}{\sqrt{T}\log(nT)}},0\right)$-differentially private. 
By the advanced composition of differential privacy, i.e., \thmref{thm:adaptive:queries}, the transcript $\calT$ is $\left(\O{1},\frac{1}{\poly(n)}\right)$-differentially private with respect to the algorithms $\calA_1,\ldots,\calA_\Gamma$.

\algref{alg:median-reg} runs $\Gamma$ instances of the data structure with error parameter $\eps$. For any given round $i \in [T]$, we say that an instance $j \in [\Gamma]$ is correct if its output $f_{i,j}$ is within a $(1\pm \eps)$ factor of $F(\bA, \bb^{(i)})$ and incorrect otherwise. For a fixed $i$, let $Y_{j}$ be the indicator variable for whether $f_{i,j}$ is correct. 

From the generalization properties of differential privacy, i.e., \thmref{thm:generalization}, we have that for any fixed iteration $i$,
\[\PPr{\left|\frac{1}{\Gamma}\sum_{j\in[\Gamma]}Y_j -\Ex{Y}\right|\ge\frac{1}{10}}<\frac{1}{\poly(m,T)},\]
where $Y$ denotes the indicator random variable for whether a random instance of the algorithm $\calA$ (not necessarily restricted to the $m$ instances maintained by the algorithm) is correct at the given round $i$. 
Since a random instance $\calA$ has randomness that is independent of the adaptive update, then $\Ex{Y}\ge\frac{3}{4}$.
Therefore, by a union bound over all $T$ rounds, we have
\[\PPr{\frac{1}{\Gamma}\sum_{i\in[\Gamma]} Y_i>0.6}>1-\frac{1}{\poly(m,T)},\]
which implies that the output on the $i$th round  is correct  with probability at least $1-\frac{1}{\poly(m,T)}$, since $T=d$. 
Then by a union bound over $i\in[T]$ for all $T$ rounds within an epoch, we have that the data structure answers all $T$ queries with probability $1-\frac{1}{m^2}$, under the adaptively chosen updates.
Finally, by a union bound over all $m$ updates, we have that the algorithm succeeds with probability at least $1-\frac{1}{m}$. 



We now analyze the run-time of the algorithm. 
The preprocessing time follows from the guarantee of \autoref{lem:dynsamp}.
For update time, we amortize over each epoch.  Within an epoch, we invoke $\Gamma=\O{\sqrt{T}\log(nT)}$ copies of the data structure in \secref{sec:dynamic-ob-reg}, and so we consider the preprocessing and update time from there and amortize over the epoch length $T$.
By \autoref{lem:preproc-time}, each copy takes    $ \beta= \O{\sqrt{\kappa (\bA)} \nnz(\bA)\log\frac{1}{\eps} +\frac{\nnz(\bA)}{\eps^2}\log n+\frac{d}{\eps^2}\log n}$ time to pre-process. 
For every step of update, each copy takes $\O{\frac{K }{\eps^{2}}\log n}$  time by \autoref{lem:update-time}. 
Therefore, the amortized update time for every epoch of length $T=\O{\frac{\nnz(\bA)}{\eps^2 K}}$ is
\begin{align*}
\O{ \frac{1}{T}  \left( \Gamma\beta +  \Gamma T\left(\frac{K}{\eps^2}\log n\right)  \right)  }     = \tO{  \sqrt{K\nnz(\bA) }  \left(\sqrt{\kappa(\bA)} + \eps^{-3}\right)  }.
\end{align*}
This completes the proof.
\end{proof}

 \subsection{An Exact and Deterministic Algorithm}
We now give a simple deterministic algorithm for the dynamic regression problem based on an SVD  trick. 
Let  $\bA = \bU \bSig\bV^\top$ be the SVD of $\bA$, where $\bU \in \R^{n\times d}, \bSig \in \R^{d\times d}$ and $ \bV\in\R^{d\times d}$.
The starting observation is that for any solution vector  $\bx$, we can write the regression cost as
\begin{align}\label{eqn:svd-trick}
    \left\|\bA \bx -\bb\right\| = \left\|\bU \bSig\bV^\top \bx -\bb\right\| = \left\|\bSig\bV^\top \bx - \bU^\top\bb\right\|,
\end{align}
 since $\bU$ is orthonormal. 
 The goal is the maintain the solution vector $ \bx = \bA^{\dagger} \bb$ and the associated right-side quantity $\left\|\bSig\bV^\top \bx -  \bU^\top\bb\right\|$.
 
Now suppose we compute $\bA^{\dagger} \in \R^{d\times n}$ and $\bU^\top \in \R^{d\times n}$ in pre-processing, and $\bA^{\dagger} \bb^{(1)}$ and $\bU^\top \bb^{(1)}$ in the first round. Then since  all subsequent updates to $\bb$ are all $K$-sparse, we only pay $\mathcal{O}(dK)$ time per step to maintain $\bA^{\dagger} \bb^{(i)}$ and $\bU^\top \bb^{(i)}$.
 
 \begin{algorithm}[!htb]
\caption{A simple SVD-based algorithm for dynamic regression}
\alglab{alg:det-reg1}
\begin{algorithmic}[1]
\Require{Design matrix $\bA \in \R^{n\times d}$, its pseudoinverse $\bA^{\dagger} \in \R^{d \times n}$ and its SVD $\bA = \bU \bSig\bV^\top$,  a sequence of  labels $\bb^{(i)}\in \R^{n}$} 
\State{Compute and store SVD $\bA = \bU \bSig\bV^\top$, where $\bU \in \R^{n\times d}, \bSig \in \R^{d\times d}, \bV\in\R^{d\times d}$}
\State{Compute and store $\bA^{\dagger}$ from the SVD.} 
\Comment{In the $1$st-round, compute  and store $\bA^{\dagger}\bb^{(1)}, \bU^\top \bb^{(1)}$.}
\For{each update $\bb^{(i)}$}
\State{Update and store $\bx^{(i)} = \bA^\dagger \bb^{(i)}$}
\State{Update and store $\bU^\top \bb^{(i)}$}
\State{Output $ F_i = \left\|\bSig\bV^\top \bx^{(i)} - \bU^\top \bb^{(i)} \right\|_2^2$}
\EndFor
\end{algorithmic}
\end{algorithm}

The algorithm is formally given  by \algref{alg:det-reg1}. Observe that the algorithm always maintains the exact  optimal regression cost. Moreover, the procedure does not require any randomness, and therefore it is adversarially robust to adaptive inputs. We formally claim the following   guarantees of the algorithm.

\begin{theorem}[Deterministic maintenance of regression costs]
     Given  $\bA, \bb^{(1)}$ and a stream of  adaptively chosen, $K$-sparse updates to the  label, \algref{alg:det-reg1} takes $\mathcal{O}(dK)$ time to update and maintain the  exact   regression cost $F(\bA,\bb^{(i)})$ at all iterations $i$. The pre-processing requires an SVD of $\bA$, in $\mathcal{O}(n^2d)$ time.
\end{theorem}

\section{Adaptive Distance Estimation}\seclab{sec:adaptive_distance_estimation}
In the adaptive distance estimation problem, there exists a set $X=\{\bx^{(1)},\ldots,\bx^{(n)}\}$ of $n$ points in $\mathbb{R}^d$. 
Given an accuracy parameter $\eps>0$, the goal is to output a $(1+\eps)$-approximation to $\|\bx^{(i)}-\bq\|_p$ for each query $\bq$ across all points $\bx^{(i)}\in X$, while minimizing the space, query time, or pre-processing time for the corresponding data structures. 
The trivial solution stores all $n$ points and computes all $n$ distances to each query point and thus can handle an unlimited number of queries. 
Since each point has dimension $d$, the trivial solution uses space and query time $\O{nd}$. 
\cite{CherapanamjeriN20} first improved the query time to $\tO{\frac{n+d}{\eps^2}}$ at the cost of using $\tO{\frac{(n+d)d}{\eps^2}}$ space and $\tO{\frac{nd^2}{\eps^2}}$ pre-processing time. 
Like the trivial solution, the algorithm of \cite{CherapanamjeriN20} also permits an arbitrary number of queries. 

In this section, we first apply our framework to show a data structure that can handle $Q$ queries of approximate distances from a \emph{specified point in $X$}, using query time $\tO{\frac{n+d}{\eps^2}}$, pre-processing time $\tO{\frac{nd\sqrt{Q}}{\eps^2}}$, and space $\tO{\frac{(n+d)\sqrt{Q}}{\eps^2}}$. 
Hence for $d\gg n\sqrt{Q}$, our data structure already improves on the work of \cite{CherapanamjeriN20}. 

However in this setting, each of the $Q$ queries returns only the approximate distance between a query point and a single point in $X$. 
By comparison, \cite{CherapanamjeriN20} outputs approximate distances to all points in $X$ and moreover, follow-up work by \cite{CherapanamjeriN22} improved the pre-processing time to $\tO{\frac{nd}{\eps^2}}$. 
Therefore, we address these two shortcomings of our framework by giving a data structure that (1) handles the case where we return the approximate distances of all points in $X$ from $Q$ adaptive query points and (2) achieves pre-processing time $\tO{\frac{nd}{\eps^2}}$. 

\begin{algorithm}[!htb]
\caption{Adaptive Distance Estimation}
\alglab{alg:ade:fast:jl}
\begin{algorithmic}[1]
\State{$r\gets\O{\sqrt{Q}\log^2(nQ)}$, $k\gets\O{\log(nQ)}$}
\State{Let $\Pi_1,\ldots,\Pi_r\in\mathbb{R}^{m\times d}$ be a JL transformation matrix (see \thmref{thm:jl} or \thmref{thm:fast:jl})}
\For{$j\in[r]$}
\State{Compute $\Pi_j\bx_i$}
\EndFor
\For{each query $(\by,i)$ with $\by\in\mathbb{R}^d$, $i\in[n]$}
\Comment{Adaptive queries}
\State{Let $S$ be a set of $k$ indices sampled (with replacement) from $[r]$}
\For{$j\in[k]$}
\State{$d_{i,j}\gets\|\Pi_{S_j}(\bx_i-\by)\|_2$}
\EndFor
\State{$d_i\gets\privmed(\{d_{i,j}\}_{j\in[k]})$, where $\privmed$ is $(1,0)$-DP.}
\State{\Return $d_i$}
\EndFor
\end{algorithmic}
\end{algorithm}



The following proof can simply be black-boxed into \thmref{thm:main:framework} using the techniques of \cite{HassidimKMMS20,BeimelKMNSS22,AttiasCSS23}. 
For completeness, we include the proof in the appendix. 
\begin{restatable}{theorem}{thmaderobust}
\thmlab{thm:ade:robust}
With high probability, we have 
\[(1-\eps)\|\bx_{i_q}-\by_q\|_2\le d_i\le(1+\eps)\|\bx_{i_q}-\by_q\|_2,\]
for all $q\in[Q]$. 
\end{restatable}

\begin{theorem}
There exists an algorithm that answers $Q$ adaptive distance estimation queries within a factor of $(1+\eps)$. 
For $\O{\left(\frac{\log d}{\eps^2}+d\log d\right)\log(nQ)}$ query time, it stores $\O{\frac{n\sqrt{Q}\log^3(nQ)}{\eps^2}}$ words of space. 
For $\O{\frac{d}{\eps^2}\log(nQ)}$ query time, it stores $\O{\frac{n\sqrt{Q}\log^2(nQ)}{\eps^2}}$ words of space. 
\end{theorem}
\begin{proof}
By \thmref{thm:fast:jl}, each fast JL transform uses $\O{\frac{\log d}{\eps^2}+d\log d}$ runtime and stores $m=\O{\frac{\log d}{\eps^2}}$ rows. 
On the other hand, by \thmref{thm:jl}, each JL transform uses $\O{\frac{d}{\eps^2}+d\log d}$ runtime and stores $m=\O{\frac{\log d}{\eps^2}}$ rows. 
\end{proof}
By comparison, \cite{CherapanamjeriN20} uses $\O{\frac{nd\log n}{\eps^2}}$ words of space and $\O{\frac{d}{\eps^2}}$ query time. 

\subsection{Faster Pre-processing Time for Adaptive Distance Estimation}
\seclab{sec:ade_srht}
In this section, we present an improved algorithm for Adaptive Distance Estimation, which allows the release of distances to \emph{all} $n$ points in the dataset for a single query, matching the query time of \cite{CherapanamjeriN20} with an improved space complexity of $\O{\eps^{-2} \sqrt{Q} n}$. 
Our results utilize a class of structured randomized linear transformations based on Hadamard matrices recursively defined below:
\begin{gather*}
    H_1 = 
    \begin{bmatrix}
        1
    \end{bmatrix}
    \qquad 
    H_d = 
    \begin{bmatrix}
        H_{d / 2} & H_{d / 2} \\
        H_{d / 2} & -H_{d / 2}
    \end{bmatrix}.
\end{gather*}
The associated class of randomized linear transformations are now defined below:
\begin{gather*}
    \{D^j\}_{j \in [m]} \subset \R^{d \times d} \text{ s.t } D^j_{k,l} \overset{iid}{\thicksim} 
    \begin{cases}
        \mc{N} (0, I) & \text{if } k = l \\
        0 & \text{otherwise}
    \end{cases} \\
    \forall z \in \R^d: h(z) = 
    \begin{bmatrix}
        H_d D^1 \\
        H_d D^2 \\
        \vdots \\
        H_d D^m \\
    \end{bmatrix} \cdot z \tag{SRHT} \label{eq:srht}.
\end{gather*}

Note that for any vector $z$, $h(z)$ may be computed in time $\O{md \log d}$ due to the recursive definition of the Hadamard transform. We now let $\phi$ and $\Phi$ denote the pdf and cdf of a standard normal random variable, $\quant_{\alpha} (\{a_i\}_{i \in [l]})$ the $\alpha^{th}$ quantile of a multi-set of real numbers $\{a_i\}_{i \in [l]}$ for any $l \in \N$ and define $\psi_r$ as follows:
\begin{equation*}
    \forall r > 0, a \in \R: \psi_r (a) \coloneqq 
    \min (\abs{a}, r).
\end{equation*}
Through the remainder of the section, we condition on the event defined in the following lemma:
\begin{lemma}[Claims 5.1 and 5.2 \cite{CherapanamjeriN22}]
    \label{lem:det_had_ade}
    For any $\delta \in \lprp{0, \frac{1}{2}}$, with probability at least $1 - \delta$:
    \begin{gather*}
        \forall z \text{ s.t } \norm{z} = 1: 2 \leq \quant_{\alpha - \beta / 4} \lprp{\{h(z)_i\}_{i \in [md]}} \leq \quant_{\alpha + \beta / 4} \lprp{\{h(z)_i\}_{i \in [md]}} \leq 4 \\
        \forall z \text{ s.t } \norm{z} = 1, r \geq 4 \sqrt{\log (1 / \eps)}: \lprp{1 - \frac{\eps}{2}} \leq \frac{1}{md} \cdot \sqrt{\frac{\pi}{2}} \cdot \sum_{i \in [md]} \psi_r (h_i (z)) \leq \lprp{1 + \frac{\eps}{2}}
    \end{gather*}
    as long as $m \geq C \eps^{-2} \log (2 / \delta) \log^5 (d/\eps) $ for some absolute constant $C > 0$. 
\end{lemma}
We will additionally require the following technical result from \cite{CherapanamjeriN22}, where for any vector $v \in \R^d$ and multiset $S = \{i_j\}_{j \in [k]}$ with $i_j \in [d]$, $v_S$ denotes the vector $[v_{i_1}, \dots, v_{i_k}]$:

\begin{lemma}[Theorem 1.4 \cite{CherapanamjeriN22}]
    \label{lem:de_alg_shrt}
    Assume $h: \R^d \to \R^{md}$ (\ref{eq:srht}) satisfies the conclusion of \lemref{det_had_ade}. Then, there is an algorithm, $\srhtadealg$, which satisfies for all $x \in \R^d$:
    \begin{equation*}
        \P_{S} \lbrb{\lprp{1 - \eps} \cdot \norm{x} \leq \srhtadealg (h(x)_S) \leq \lprp{1 + \eps} \cdot \norm{x}} \geq 1 - \delta \text{ for } S = \{i_j\}_{j \in [k]} \text{ with } i_j \overset{iid}{\thicksim} \unif ([md])
    \end{equation*}
    when $k \geq C \eps^{-2} \log (2 / \eps) \log (2 / \delta)$ for some $C > 0$. Furthermore, $\srhtadealg$ runs in time $\O{k}$. 
\end{lemma}
With these primitives, we will construct our data structure for adaptive distance estimation. Our constructions is formally described in \algref{ade_srht}. 

\begin{algorithm}[H]
\caption{Adaptive Distance Estimation with SRHTs}
\alglab{ade_srht}
\begin{algorithmic}[1]
    \State $m \gets C \eps^{-2} \log^6 (2dn / \eps)$
    \State Let $h$ be an SRHT as defined in \ref{eq:srht}
    \Comment{Revealed to analyst}
    \State $r \gets C \sqrt{Q} \log^3(nd)$, $k \gets C \eps^{-2} \log (2 / \eps) \log (2 nd)$
    \For{$i \in [n]$}
        \State{Compute $y_i = h(x_i)$}
        \For {$j \in [r]$}
            \State Let $S_{i, j}$ be a set of $k$ indices sampled with replacement from $[md]$
        \EndFor
    \EndFor
    
    \State $l \gets C \log(nd)$
    \For{$j \in 1:Q$}
        \Comment{Adaptive queries}
        \State Receive query $q_j$
        \State $v_j \gets h(q_j)$
        \For {$i \in [n]$}
            \State{Let $\{t_{i, j, p}\}_{p \in [l]}$ be a set of $l$ indices sampled (with replacement) from $[r]$}
                \For{$p \in [l]$}
                    \State{$d_{i, j, p} \gets \srhtadealg ((v_j - y_i)_{S_{i, t_{i, j, p}}})$}
                \EndFor
            \State{$d_{i, j} \gets \privmed(\{d_{i, j, p}\}_{p \in [l]})$, where $\privmed$ is $(\O{1},0)$-DP.}
        \EndFor
        \State{\Return $\{d_{i, j}\}_{i \in [n]}$}
    \EndFor
\end{algorithmic}
\end{algorithm}

The proof of correctness of \algref{ade_srht} will follow along similar lines to that of \algref{alg:framework} with a more refined analysis of the privacy loss incurred due to the adaptivity of the data analyst. In particular, each input query results in $n$ different queries made to a differentially private mechanism $\privmed$ leading to a total of $nQ$ queries. A na\"ive application of \thmref{thm:main:framework} would thus result in a data structure with space complexity scaling as $\wt{O}(n^{3 / 2} \sqrt{Q})$ as opposed to the desired $\wt{O} (n \sqrt{Q})$ and query complexity $\wt{O}(\eps^{-2} nd)$. The key insight yielding the improved result is the privacy loss incurred by a single query is effectively amortized across $n$ independent differentially private algorithms each capable of answering $Q$ adaptively chosen queries correctly with high probability. 

To start, we first condition on the event in \lemref{det_had_ade} and assume public access to the correspondingly defined SRHT $h$. We now use $R$ to denote the randomness used to instantiate the multisets, $S_{i, j}$, in \algref{ade_srht} and decompose it as follows $R = \{R_{i}\}_{i \in [n]}$ with $R_i = \{R_{i, j}\}_{j \in [r]}$ where $R_{i, j}$ corresponds to the randomness used to generate the set $S_{i, j}$ and the random elements $t_{i, p}$. As in the proof of \thmref{thm:main:framework}, we define a transcript $T = \{T_j\}_{j \in [Q]}$ with $T_j = (q_j, \{d_{i, j}\}_{i \in [n]})$ denoting the $j^{th}$ query and the responses returned by \algref{ade_srht} as a single transaction. 

We now prove the correctness of our improved procedure for adaptive distance estimation. 

\begin{proofof}{\thmref{ade_srht}}
    We condition on the event in the conclusion of \lemref{det_had_ade} start by bounding the failure probability of a single query. The bound for the whole sequence of adaptively chosen queries follows by a union bound. Now, fixing $i \in [n]$ and $j \in [Q]$, note that the sub-transcript $T^{(j)} = \{T_p\}_{p \in [j - 1]}$ is $\lprp{o(1), \frac{1}{\poly (nQ)}}$-differentially private with respect to $R_i$. Furthermore, define the indicator random variables:
    \begin{align*}
        \forall p \in [l]: W_p &\coloneqq \bm{1} \lbrb{(1 - \eps) \cdot \norm{q_j - x_i} \leq \srhtadealg \lprp{(v_j - y_i)_{S_{i, t_{i, j, p}}}} \leq (1 + \eps) \cdot \norm{q_j - x_i}}
    \end{align*}
    Additionally, defining $W \coloneqq \sum_{p = 1}^l W_p$, we get by the differential privacy of the sub-transcript, $T^{(j)}$, \lemref{de_alg_shrt} and \thmref{thm:generalization}:
    \begin{align*}
        \P \lbrb{W \leq \frac{3}{4} \cdot l} \leq \frac{1}{400 \cdot (nQ)^2}.
    \end{align*}
    Consequently, we get from \thmref{thm:dp:median} and another union bound:
    \begin{equation*}
        \P \lbrb{(1 - \eps) \cdot \norm{q_j - x_i} \leq d_{i, j} \leq (1 + \eps) \cdot \norm{q_j - x_i}} \geq 1 - \frac{1}{200 \cdot (nQ)^2}.
    \end{equation*}
    A subsequent union bound over all $i \in [n], j \in [Q]$ yields:
    \begin{equation*}
        \P \lbrb{\forall i \in [n], j \in [Q]: (1 - \eps) \cdot \norm{q_j - x_i} \leq d_{i, j} \leq (1 + \eps) \cdot \norm{q_j - x_i}} \geq 1 - \frac{1}{200 \cdot (nQ)}.
    \end{equation*}
    A final union bound over the conclusion of \lemref{det_had_ade} concludes the proof. The runtime guarantees follow from the fact that for all $z \in \R^d$, $h(z)$ is computable in time $\O{md\log d}$ and the runtime guarantees of $\srhtadealg$. 
\end{proofof}

\section{Adaptive Kernel Density Estimation}
Kernel density estimation is an important problem in learning theory and statistics that has recently attracted significant interest, e.g., \cite{CharikarS17,BackursCIS18,CharikarKNS20,BakshiIKSZ22}. 
In the adaptive kernel density estimation problem, the input is a set $X=\{\bx^{(1)},\ldots,\bx^{(n)}\}$ of $n$ points in $\mathbb{R}^d$. 
Given an accuracy parameter $\eps>0$ and a threshold parameter $\tau>0$, the goal is to output a $(1+\eps)$-approximation to the quantity $\frac{1}{n}\sum_{i\in[n]}k(\bx^{(i)},\bq)$, for a kernel function $k$ under the promise that the output is at least $\tau$. 
A standard approach is to sample $\O{\frac{1}{\tau\eps^2}}$ points and then use $\O{\frac{d}{\tau\eps^2}}$ query time to output the empirical kernel density for a specific query. 
\cite{BackursIW19} give an algorithm for kernel density estimation that uses $\O{\frac{1}{\tau\eps^2}}$ space and $\O{\frac{d}{\sqrt{\tau}\eps^2}}$ query time, improving over the standard sampling approach. 

\begin{theorem}
\thmlab{thm:kde:better}
\cite{BackursIW19}
Given $\eps,\tau>0$, there exists a data structure $D$ that uses $\O{\frac{1}{\tau\eps^2}}$ space and $\O{\frac{d}{\eps^2\sqrt{\tau}}}$ query time that outputs a $(1+\eps)$-approximation $D(\by)$ to a kernel density estimation query $\by$ that has value at least $\tau$, i.e., 
\[\PPr{|D(\by)-\kde(X,\by)|\le\eps\cdot\kde(X,\by)}\ge\frac{3}{4}.\]
\end{theorem}

However, the analysis for both these algorithms fails for the adaptive setting, where there can be dependencies between the query and the data structure. 
By using the data structure of \cite{BackursIW19} as a subroutine, our framework immediately implies an algorithm for adaptive kernel density estimation that uses $\tO{\frac{\sqrt{Q}}{\tau\eps^2}}$ space and $\O{\frac{d\log Q}{\sqrt{\tau}\eps^2}}$ query time to answer each of $Q$ adaptive queries. 

\begin{algorithm}[!htb]
\caption{Adaptive Kernel Density Estimation}
\alglab{alg:kde}
\begin{algorithmic}[1]
\Require{Number $Q$ of queries, accuracy $\eps$, threshold $\tau$}
\State{$r\gets\O{\sqrt{Q}\log^2 Q}$}
\For{$i\in[r]$}
\Comment{Pre-processing}
\State{Let $T_i$ be a KDE data structure}
\EndFor
\For{each query $\by_q\in\mathbb{R}^d$ with $q\in[Q]$}
\Comment{Adaptive queries}
\State{Let $S$ be a set of $k$ indices sampled (with replacement) from $[r]$}
\For{$i\in[k]$}
\State{Let $D_i$ be the output of $T_{S_i}$ on query $\by_q$}
\EndFor
\State{\Return $d_q=\privmed(\{D_i\}_{i\in[k]})$, where $\privmed$ is $(1,0)$-DP.}
\EndFor
\end{algorithmic}
\end{algorithm}

We first claim adversarial robustness of our algorithm across $Q$ adaptive queries. 
Since the proof can simply be black-boxed into \thmref{thm:main:framework} using the techniques of \cite{HassidimKMMS20,BeimelKMNSS22,AttiasCSS23}, we defer the proof of the following statement to the appendix.
\begin{restatable}{lemma}{lemkderobust}
\lemlab{lem:kde:robust}
\algref{alg:kde} answers $Q$ adaptive kernel density estimation queries within a factor of $(1+\eps)$, provided each query has value at least $\tau$. 
\end{restatable}

\begin{theorem}
There exists an algorithm that uses $\O{\frac{\sqrt{Q}\log^2 Q}{\tau\eps^2}}$ space and answers $Q$ adaptive kernel density estimation queries within a factor of $(1+\eps)$, provided each query has value at least $\tau$. 
Each query uses $\O{\frac{d\log(nQ)}{\eps^2\sqrt{\tau}}}$ runtime.  
\end{theorem}
By comparison, random sampling, e.g.,~\cite{CharikarS17}, uses $\frac{Q}{\tau\eps^2}$ samples to answer $Q$ queries and each query uses $\O{\frac{d}{\tau\eps^2}}$ runtime and using $Q$ copies of the data structure by \cite{BackursIW19} uses $\O{\frac{Q}{\tau\eps^2}}$ space and $\O{\frac{d}{\eps^2\sqrt{\tau}}}$ runtime. 

\subsection{Unlimited Adaptive Queries for Kernel Density Estimation}
In this section, we go beyond the limits of our framework and analyze the case where there may be an unbounded number of adversarial queries. 

\thmakdeinfty*
\begin{proof}
Given a set $X\subseteq\mathbb{R}^d$ of $n$ points with diameter $\Delta$, let $\calN$ be an $\frac{\eps\tau}{L}$-net over a ball of radius $\Delta+\rho$ that contains $X$. 
More formally, let $B$ be a ball of radius $(\Delta+\rho)$ that contains $X$ and for every $\by\in B$, there exists a point $\bz\in\calN$ such that $\|\by-\bz\|_2\le\frac{\eps\tau}{L}$. 
We can construct the net greedily so that $|\calN|\le\left(\frac{2(\Delta+\rho) L}{\eps\tau}\right)^d$. 

We implement a data structure $D$ that answers each (non-adaptive) kernel density estimation query with multiplicative approximation $\left(1+\frac{\eps}{3}\right)$ for any kernel density estimation query with value at least $\frac{\tau}{2}$, with probability at least $1-\delta$, where $\delta\le\frac{1}{100|\calN|}$. 
Then by a union bound, $D$ correctly answers each kernel density estimation query in $\calN$ with probability at least $0.99$. 

Let $\bq\in\mathbb{R}^d$ be an arbitrary query such that $\kde(X,\bq)\ge\tau$. 
By assumption, we have that $\|\bq-\bx\|_2\le\rho$ for some $\bx\in X$ and thus $\bq\in B$. 
By the definition of $\calN$, there exists some $\by\in\calN$ such that $\|\bq-\by\|_2\le\frac{\eps\tau}{3L}$. 
Then since $k$ is $L$-Lipschitz in the second variable, we have
\[|\kde(X,\bq)-\kde(X,\by)|=\left|\frac{1}{n}\sum_{\bx\in X}k(\bx,\bq)-\frac{1}{n}\sum_{\bx\in X}k(\bx,\by)\right|\le\frac{L}{n}\|\bq-\by\|_2\le\frac{\eps\tau}{3n}.\]
Hence, $\kde(X,\bq)\ge\tau$ implies that $\kde(X,\by)\ge\frac{\tau}{2}$. 
Let $K_\by$ be the output of the data structure $D$ on query $\by$. 
Then by correctness of $D$ on $\calN$ for any query with threshold at least $\frac{\tau}{2}$, we have 
\[\left|K_\by-\kde(X,\by)\right|\le\frac{\eps}{3}\kde(X,\by).\]
Let $K_\bq$ be the output of the data structure $D$ on query $\by$. 
Since the algorithm itself is $L$-Lipschitz, then
\[|K_\bq-K_\by|\le L\|\bq-\by\|_2\le\frac{\eps\tau}{3}.\]
Therefore by the triangle inequality, we have that
\begin{align*}
|K_\bq-\kde(X,\bq)|&\le|K_\bq-K_\by|-|K_\by-\kde(X,\by)|-|\kde(X,\by)-\kde(X,\bq)|\\
&\le\frac{\eps\tau}{3}+\frac{\eps}{3}\kde(X,\by)+\frac{\eps\tau}{3n}.
\end{align*}
Since $\kde(X,\by)\le\kde(X,\bq)+\frac{\eps\tau}{3n}$, then it follows that
\[|K_\bq-\kde(X,\bq)|\le\frac{\eps\tau}{3}+\frac{\eps}{3}\kde(X,\bq)+\frac{\eps^2\tau}{n}+\frac{\eps\tau}{3n}\le\eps\kde(X,\bq),\]
for $n\ge 6$.
\end{proof}

In particular, sampling-based algorithms for kernels that are Lipschitz are also Lipschitz. 
Thus to apply \thmref{thm:akde:infty}, it suffices to identify kernels that are $L$-Lipschitz and use the data structure of  \thmref{thm:kde:better}. 
To that end, we note that the kernels $k(\bx,\by)=\frac{C}{C+\|\bx-\by\|_2}$ for $C>0$ and $k(\bx,\by)=Ce^{-\|\bx-\by\|_2}$ are both Lipschitz for some function of $C$. 
In particular, we have
\begin{align*}
|k(\bx,\by)-k(\bx,\bz)|&=\left|\frac{C}{C+\|\bx-\by\|_2}-\frac{C}{C+\|\bx-\bz\|_2}\right|\\
&=\frac{C|\|\bx-\bz\|_2-\|\bx-\by\|_2|}{(C+\|\bx-\by\|_2)(C+\|\bx-\bz\|_2)}\\
&\le\frac{\|\by-\bz\|_2}{C},
\end{align*}
so $k(\bx,\by)=\frac{C}{C+\|\bx-\by\|_2}$ is $\frac{1}{C}$-Lipschitz. 
Similarly, since $e^{-x}$ is $1$-Lipschitz, then 
\begin{align*}
|k(\bx,\by)-k(\bx,\bz)|&=Ce^{-\|\bx-\by\|_2}-Ce^{-\|\bx-\bz\|_2}\\
&\le C|\|\bx-\bz\|_2-\|\bx-\by\|_2|\le C\|\by-\bz\|_2,
\end{align*}
so $k(\bx,\by)=Ce^{-\|\bx-\by\|_2}$ is $C$-Lipschitz. 

\section{Empirical Evaluation}
We empirically demonstrate the space and query time efficiency of our approach of \secref{sec:adaptive_distance_estimation}. We consider the problem of $\ell_2$ norm estimation where queries $q_1, q_2, \ldots$ are generated in an \emph{adaptive} fashion and our goal is to output an estimate of $\|q_i\|_2$ for all $i$. This setting is a special case of adaptive distance estimation and captures the essence of our adversarial robustness framework. In addition, this same setting was investigated empirically in prior works \cite{CherapanamjeriN20}. 

\paragraph{Experimental Setup.} Consider the setting of \algref{alg:framework}: it creates $r$ copies of an underlying randomized data structure and upon a query, it subsamples $k$ of them and outputs an answer aggregated via the private median. In our setting, the underlying algorithm will be the \emph{fast Johnson-Lindenstrauss} (JL) transform which is defined as follows: it is the matrix $PHD: \R^{d} \rightarrow \R^{m}$ where $D$ is a diagonal matrix with uniformly random $\pm 1$ entries, $H$ is the Hadamard transform, and $P$ is a sampling matrix uniformly samples $m$ rows of $HD$. Our algorithm will initialize $r$ copies of this matrix where the sampling matrix $P$ and diagonal $D$ will be the randomness which is ``hidden" from the adversary. Upon query $q$, we sample $k$ different Fast JL data structures, input $q$ to all of them, and proceed as in \algref{alg:framework}. 
Note that this setting exactly mimics the theoretical guarantees of \secref{sec:framework} and is exactly \algref{alg:ade:fast:jl} of \secref{sec:adaptive_distance_estimation}. 
In our experiments, $d = 4096, m = 250, r = 200,$ and $k=5$. These are exactly the parameters chosen in prior works \cite{CherapanamjeriN20}.  We will have $5000$ adaptive queries $q_i$ which are described shortly. Our experiments are done on a 2021 M1 Macbook Pro with 32 gigabytes of RAM. We implemented all algorithms in Python 3.5 using Numpy. The Hadamard transform code is from \cite{andoni2015practical}\footnote{available in \url{https://github.com/FALCONN-LIB/FFHT}} and we use Google's differential privacy library\footnote{available in \url{https://github.com/google/differential-privacy}} for the private median implementation.

\paragraph{Baselines.} We will consider three baselines. \textbf{JL} will denote a standard (Gaussian) JL map from dimension $4096$ to $250$. \textbf{Baseline 1} will denote the algorithm of \cite{CherapanamjeriN20}. At a high level, it instantiates many independent copies of the standard Gaussian JL map and only feeds an incoming query into a select number of subsampled data structures. Note that our experimental setting is mimicking exactly that of \cite{CherapanamjeriN20} where the same parameters $r$ (number of different underlying data structures) and $k$ (number of subsampled data structures to use for a query) were used. This ensures that both our algorithm and theirs have access to the same number of \emph{different} JL maps and thus allows us to compare the two approaches on an equal footing. The last baseline, denoted as \textbf{Baseline 2}, is the main algorithm of \cite{CherapanamjeriN22} which is the optimized version of \cite{CherapanamjeriN20}. At a high level, their algorithm proceeds similarly to that of \cite{CherapanamjeriN20}, except they employ Hadamard transforms (after multiplying the query entry-wise by random Gaussians), rather than using Gaussian JL maps. Furthermore, instead of subsampling, their algorithm feeds an incoming query into all the different copies of the Hadamard transform, and subsamples the coordinates of the  concatenated output for norm estimation. We again set the parameters of their algorithm to match that of our algorithm and \textbf{Baseline 1} by using $r$ copies of their Hadamard transform and subsampling $mk$ total coordinates. We refer to the respective papers for full details of their algorithms.

\paragraph{Summary of adaptive queries.}
Our input queries are the same adaptive queries used in \cite{CherapanamjeriN20}. To summarize, let $\Pi$ denote the map used in the \textbf{JL} benchmark stated above. The $i$-th query for $1 \le i \le 5000$ will be of the form $q_i = \sum_{j=1}^i (-1)^{W_i} z_i$, which we then normalize to have unit norm. The $z_i$ are standard Gaussian vectors. $W_i$ is the indicator variable for the event $\|\Pi(z_i - e_1)\|_2 \le \| \Pi(z_i + e_1) \|_2$ where $e_1$ is the first standard basis vector. Intuitively, the queries become increasingly correlated with the matrix $\Pi$ since we successively ``augment" the queries in a biased fashion. See Section $5$ of \cite{CherapanamjeriN20} for a more detailed discussion of the adaptive inputs.

\paragraph{Results.}
Our results are shown in \figref{fig:synthetic}. In \figref{fig:synthetic_a}, we plot the norm estimated by each of the algorithms in each of the queries across iterations. We see that the na\"ive \textbf{JL} map increasingly deviates from the true value of $1.0$. This is intuitive as the adaptive queries are increasingly correlated with the map $\Pi$. The performance of all other algorithms are indistinguishable in \figref{fig:synthetic_a}. Thus, we only zoom into the performances of our algorithm and \textbf{Baseline 1} and \textbf{Baseline 2}, shown in \figref{fig:synthetic_b}. For these three algorithms, we plot a histogram of answers outputted by the respective algorithms across all iterations. We see that the algorithm of \cite{CherapanamjeriN20}, shown in the blue shaded histogram, is the most accurate as it has the smallest deviations from the true answer of $1.0$. Our algorithm, shown in green, is noisier than \textbf{Baseline 1} since it has a wider range of variability. This may be due to the fact that we use a differentially private median algorithm, which naturally incurs additional noise. Lastly, \textbf{Baseline 2} is also noisier than \textbf{Baseline 1} and comparable to our algorithm. This may be due to the fact that the algorithm of \cite{CherapanamjeriN22} requires very fine-tuned constants in their theoretical bounds, which naturally deviate in practice. Lastly, \figref{fig:synthetic_c} shows the cumulative runtime of all three algorithms across all iterations. Our algorithm, shown in green, is the fastest while \textbf{Baseline 2} is the slowest. This is explained by the fact that \textbf{Baseline 2} calculates many more Hadamard transforms than our algorithm does.

\begin{figure}
     \centering
     \begin{subfigure}[b]{0.32\textwidth}
         \centering
         \includegraphics[width=\textwidth]{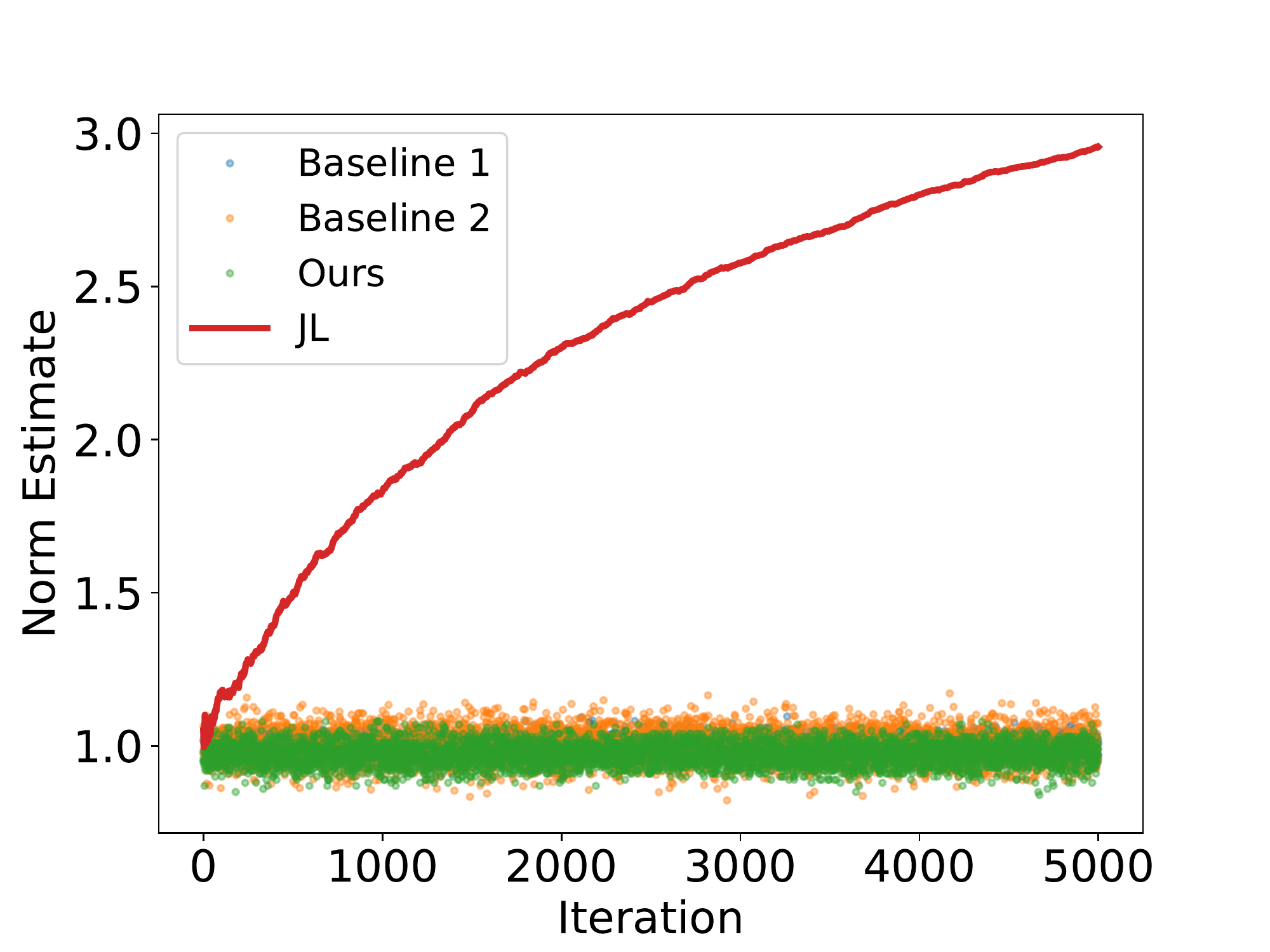}
         \caption{}
         \figlab{fig:synthetic_a}
     \end{subfigure}
     \hfill
     \begin{subfigure}[b]{0.32\textwidth}
         \centering
         \includegraphics[width=\textwidth]{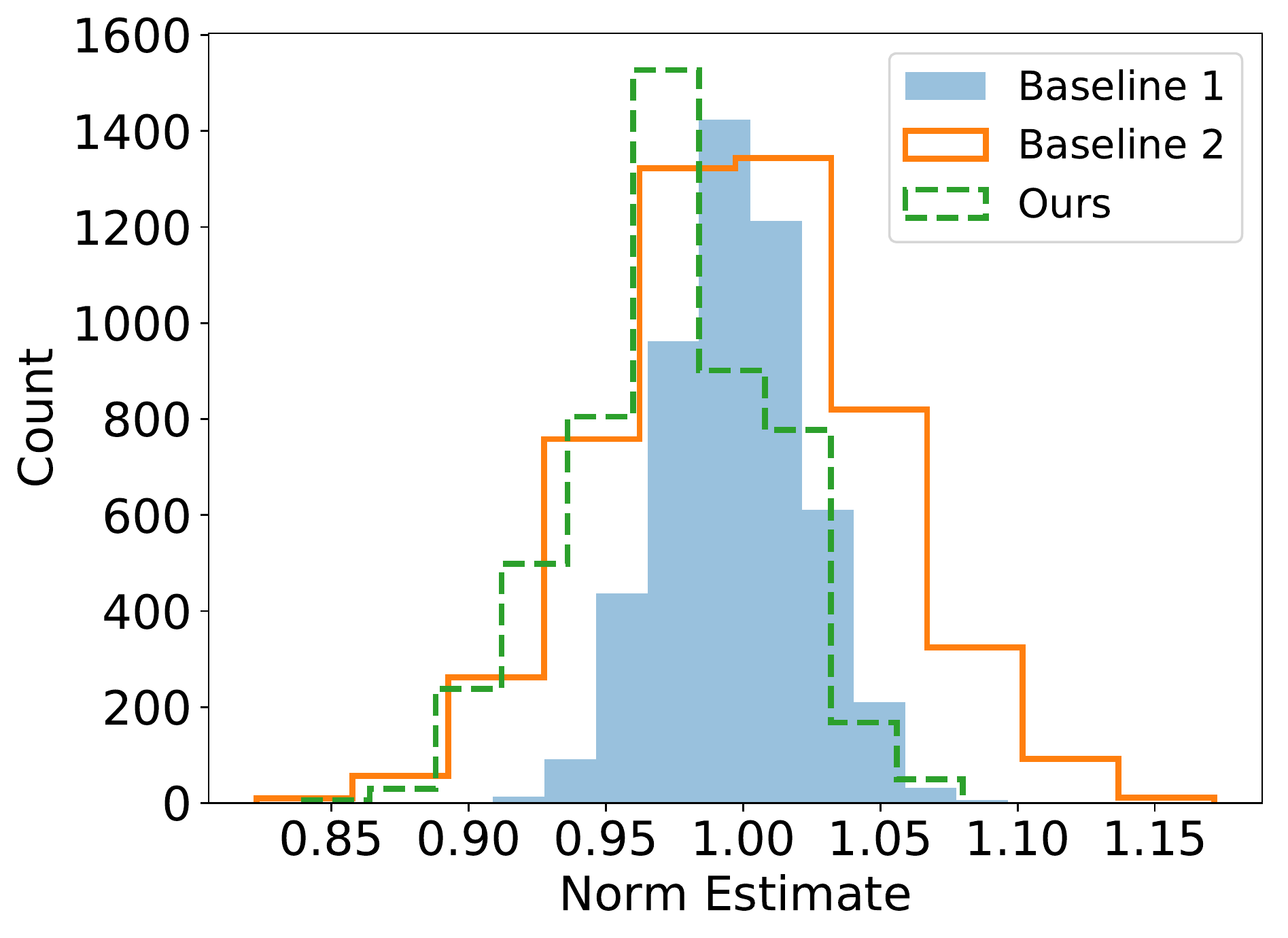}
         \caption{}
         \figlab{fig:synthetic_b}
     \end{subfigure}
     \hfill
     \begin{subfigure}[b]{0.32\textwidth}
         \centering
         \includegraphics[width=\textwidth]{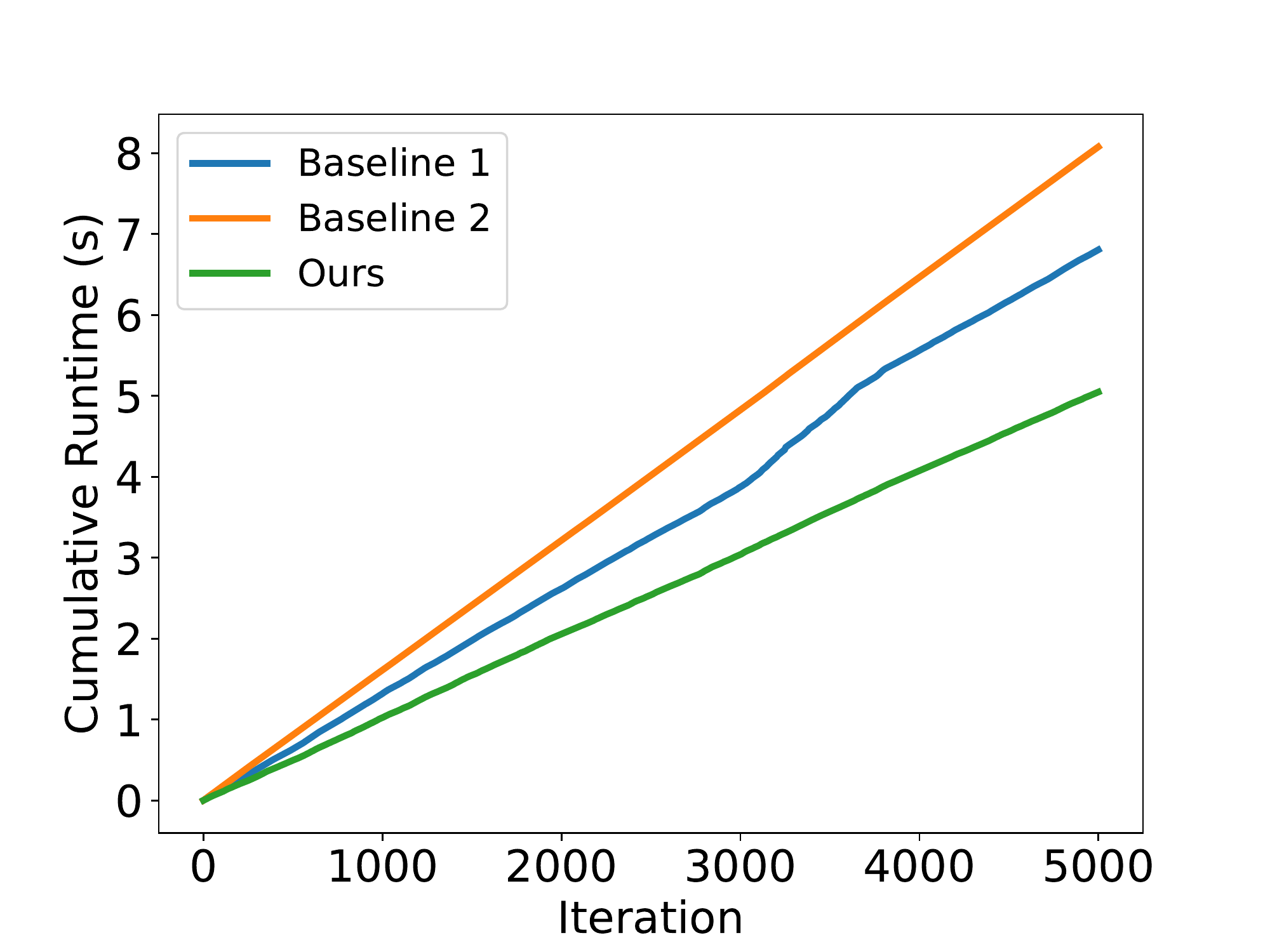}
         \caption{}
         \figlab{fig:synthetic_c}
     \end{subfigure}
        \caption{Figures for our experiments.}
        \figlab{fig:synthetic}
\end{figure}

\def\shortbib{0}
\bibliographystyle{alpha}
\bibliography{references}

\appendix

\section{A Simple Framework for Adversarial Robustness}
\seclab{sec:framework}
In this section, we describe the benchmark framework that enables $Q$ adaptive queries to a data structure by using $\tO{\sqrt{Q}}$ copies of a non-adaptive data structure. 
The framework and corresponding analysis of correctness are simply compartmentalizations of the techniques in \cite{HassidimKMMS20,BeimelKMNSS22,AttiasCSS23}. 
For the sake of completeness, we include them here and discuss additional applications. 
Namely, we show that through advanced composition of differential privacy, the private median of $\tO{\sqrt{Q}}$ copies protects the internal randomness of each non-adaptive data structure while still adding sufficiently small noise to guarantee accuracy. 
Moreover, we use amplification of privacy by sampling to only consider a small subset of the $\tO{\sqrt{Q}}$ non-adaptive data structures to further improve the runtime. 

\begin{algorithm}[!htb]
\caption{Adaptive Algorithm Interaction}
\alglab{alg:framework}
\begin{algorithmic}[1]
\State{$r\gets\O{\sqrt{Q}\log^2(nQ)}$, $k\gets\O{\log(nQ)}$}
\For{$i\in[r]$}
\State{Implement data structure $\calD_i$ on the input}
\EndFor
\For{each query $q_i$, $i\in[Q]$}
\State{Let $S$ be a set of $k$ indices sampled (with replacement) from $[r]$}
\State{For each $j\in[k]$, let $d_{i,j}$ be the output of $\calD_{S_j}$ on query $q_i$}
\State{$d_i\gets\privmed(\{d_{i,j}\}_{j\in[k]})$, where $\privmed$ is $(1,0)$-DP}
\EndFor
\end{algorithmic}
\end{algorithm}

We first argue that \algref{alg:framework} maintains accuracy against $Q$ rounds of interaction with an adaptive adversary. 
Let $R=\{R^{(0)},R^{(1)},\ldots,R^{(r)}\}$, where $R^{(1)},\ldots,R^{(r)}$ denotes the random strings used by the oblivious data structures $\calD_1,\ldots,\calD_r$ and $R^{(0)}$ denotes the additional randomness used by \algref{alg:framework}, such as in the private median subroutine $\privmed$. 
Consider a transcript $T(R)=\{T_1,\ldots,T_Q\}$ such that for each $i\in[Q]$, we define $T_i=(q_i,d_i)$ to be the ordered pair consisting of the query $q_i$ and the corresponding answer $d_i$ by \algref{alg:framework} using the random string $R^{(0)}$, as well as the oblivious data structures $\calD_1,\ldots,\calD_r$ with random strings $R^{(1)},\ldots,R^{(r)}$. 
We remark that $d_i$ is a random variable due to the randomness of each data structure, as well as the randomness of the private median subroutine $\privmed$. 

We will first argue that the transcript $T_R$ is differentially private with respect to $R$. 
We emphasize that similar arguments were made in the streaming model by~\cite{HassidimKMMS20} and in the dynamic model~\cite{BeimelKMNSS22,AttiasCSS23}. 

\begin{lemma}
\lemlab{lem:dp:single:iter}
For a fixed iteration, $T_i$ is $\left(\O{\frac{1}{\sqrt{Q}\log(nQ)}},0\right)$-differentially private with respect to $R$. 
\end{lemma}
\begin{proof}
We first observe that $\privmed$ is $(1,0)$-differentially private on the outputs of the $r=\O{\sqrt{Q}\log^2(nQ)}$ data structures. 
\algref{alg:framework} samples $k=\O{\log(nQ)}$ groups of data structures from the $r$ total data structures. 
Thus by amplification via sampling, i.e., \thmref{thm:dp:sampling}, $\privmed$ is $\left(\O{\frac{1}{\sqrt{Q}\log(nQ)}},0\right)$-differentially private. 
Therefore, $T_i$ is $\left(\O{\frac{1}{\sqrt{Q}\log(nQ)}},0\right)$-differentially private with respect to $R$.
\end{proof}

\noindent
We next argue that the entire transcript is differentially private with respect to the randomness $R$. 
\begin{lemma}
\lemlab{lem:dp:all:iter}
$T$ is $\left(\O{1},\frac{1}{\poly(nQ)}\right)$-differentially private with respect to $R$. 
\end{lemma}
\begin{proof}
By \lemref{lem:dp:single:iter}, for each fixed iteration $i\in[Q]$, the transcript $T_i$ is $\left(\O{\frac{1}{\sqrt{Q}\log(nQ)}},0\right)$-differentially private with respect to $R$. 
Note that the transcript $T$ is an adaptive composition of the transcripts $T_1,\ldots,T_Q$. 
Thus, by the advanced composition of differential privacy, i.e., \thmref{thm:adaptive:queries}, the transcript $T$ is $\left(\O{1},\frac{1}{\poly(nQ)}\right)$-differentially private with respect to $R$. 
\end{proof}

\noindent
We now prove the correctness of our unifying framework. 
\begin{proofof}{\thmref{thm:main:framework}}
For a fixed query $q_i$ with $i\in[Q]$, let $S$ be the corresponding set of $k$ indices sampled from $[r]$. 
Let $\calV$ be the set of valid answers on query $q_i$. 
Let $I_j$ be an indicator variable for whether the output $d_{i,j}$ on query $q_i$ by $\calD_{S_j}$ is correct, so that $I_j=1$ if $d_{i,j}\in\calV$ and $I_j=0$ if $d_{i,j}\notin\calV$.  
By assumption, we have that for each $j\in[k]$, 
\[\PPr{I_j=1}\ge\frac{3}{4},\]
so that $\Ex{I_j}\ge\frac{3}{4}$. 
We define the random variable $I=\frac{1}{k}\sum_{j\in[k]}I_j$ so that by linearity of expectation, $\Ex{I}=\frac{1}{k}\sum_{j\in[k]}\Ex{I_j}\ge\frac{3}{4}$. 

To handle a sequence of $Q$ adaptive queries, we consider the transcript $T(R)=\{T_1,\ldots,T_Q\}$ for the randomness $R=\{R^{(0)},R^{(1)},\ldots,R^{(r)}\}$ previously defined, i.e., for each $i\in[Q]$, $T_i=(q_i,d_i)$ is the ordered pair consisting of the query $q_i$ and the corresponding answer $d_i$ by \algref{alg:framework} using the random string $R^{(0)}$, as well as the oblivious data structures $\calD_1,\ldots,\calD_r$ with random strings $R^{(1)},\ldots,R^{(r)}$. 
By \lemref{lem:dp:all:iter}, we have that $T$ is $\left(\O{1},\frac{1}{\poly(nQ)}\right)$-differentially private with respect to $R$. 

For $j\in[k]$, we define the function $\issucc(R^{(S_j)})$ to be the indicator variable for whether the output $d_{i,S_j}$ by data structure $D_{S_j}$ is successful on query $q_i$. 
For example, if $D$ is supposed to answer queries within $(1+\alpha)$-approximation, then we define $\issucc(R^{(S_j)})$ to be one if $d_{i,S_j}$ is within a $(1+\alpha)$-approximation to the true answer on query $q_i$, and zero otherwise. 
From the generalization properties of differential privacy, i.e., \thmref{thm:generalization}, we have
\[\PPr{\left|\frac{1}{k}\sum_{j\in[k]}\issucc(R^{(S_j)}) -\EEx{\overline{R}}{\issucc(\overline{R})}\right|\ge\frac{1}{10}}<\frac{1}{\poly(n,Q)},\]
for sufficiently small $\O{1}$. 
Therefore, by a union bound over all $Q$ queries, we have
\[\PPr{\frac{1}{k}\sum_{i\in[k]} I_i>0.6}>1-\frac{1}{\poly(n,Q)},\]
which implies that $d_i$ is correct on query $q_i$. 
Then by a union bound over $i\in[Q]$ for all $Q$ adaptive queries, we have that the data structure answers all $Q$ adaptive queries with high probability. 
\end{proofof}
\thmref{thm:main:framework} has applications to a number of central problems in data science and machine learning, such as adaptive distance estimation, kernel density estimation, nearest neighbor search, matrix-vector norm queries, linear regression,  range queries, and point queries. 
In the remainder of the section, we formally describe the range queries, point queries, matrix-vector norm queries, and linear regression problems; we defer discussion of adaptive distance estimation, kernel density estimation, and nearest neighbor search to the subsequent sections. 

\subsection{Application: Matrix-Vector Norm Queries}
In the matrix-vector norm query problem, we are given a matrix $\bA\in\mathbb{R}^{n\times d}$ and we would like to handle $Q$ adaptive queries $\bx^{(1)},\ldots,\bx^{(Q)}$ for an approximation parameter $\eps>0$ by outputting a $(1+\eps)$-approximation to $\|\bA\bx^{(i)}\|_p$ for each query $\bx^{(i)}\in\mathbb{R}^d$ with $i\in[Q]$. 
Here we define $\|\bv\|_p^p=\sum_{i\in[d]}(v_i)^p$ for a vector $\bv\in\mathbb{R}^d$. 
Observe that computing $\bA\bx^{(i)}$ explicitly and then computing its $p$-norm requires $\O{\nd}$ time. 
Thus for $n\gg d$, a much faster approach is to produce a subspace embedding, i.e., to compute a matrix $\bM\in\mathbb{R}^{m\times d}$ with $m\ll n$, such that for all $\bx\in\mathbb{R}^d$,
\[(1-\eps)\|\bA\bx\|_p\le\|\bM\bx\|_p\le(1+\eps)\|\bA\bx\|_p.\]
However, because subspace embeddings must be correct over all possible queries, the number of rows of $\bM$ is usually $m=\Omega\left(\frac{d}{\eps^2}\right)$ due to requiring correctness over an $\eps$-net. 

\begin{theorem}[\cite{Indyk06,Li08}]
\thmlab{thm:p:stable}
Given $\bA\in\mathbb{R}^{n\times d}$, $p\in(0,2]$, and an accuracy parameter $\eps>0$, there exists an algorithm that creates a data structure that uses $\O{\frac{1}{\eps^2}\log n}$ bits of space and outputs a $(1+\eps)$-approximation to $\|\bA\bx\|_p$ for a query $\bx\in\mathbb{R}^d$, with high probability, in time $\O{\frac{d}{\eps^2}\log n}$.
\end{theorem}

\thmref{thm:p:stable} essentially creates a matrix $\bR\in\mathbb{R}^{m\times n}$ of random variables sampled from the $p$-stable distribution~\cite{Zolotarev86} and then stores the matrix $\bR\bA$. 
Once the query $\bx$ arrives, the data structure then outputs a $(1+\eps)$-approximation to $\|\bA\bx\|_p$ by computing a predetermined function on $\bR\bA\bx$. 
The restriction on $p\in(0,2]$ is due to the fact that the $p$-stable distributions only exist for $p\in(0,2]$. 
From \thmref{thm:p:stable} and \thmref{thm:main:framework}, we have the following:
\begin{theorem}
Given $\bA\in\mathbb{R}^{n\times d}$, $p\in(0,2]$, and an accuracy parameter $\eps>0$, there exists an algorithm that creates a data structure that uses $\O{\frac{\sqrt{Q}}{\eps^2}\log^2(nQ)}$ bits of space and outputs a $(1+\eps)$-approximation to $\|\bA\bx^{(i)}\|_p$ with $i\in[Q]$ for $Q$ adaptive queries $\bx^{(1)},\ldots,\bx^{(Q)}\in\mathbb{R}^d$, with high probability, in time $\tO{\frac{d}{\eps^2}\log^2(nQ)+\log^3(nQ)}$.
\end{theorem}

\subsection{Application: Linear Regression}
In the linear regression problem, we are given a fixed matrix $\bA\in\mathbb{R}^{n\times d}$ and we would like to handle $Q$ adaptive queries $\bb^{(1)},\ldots,\bb^{(Q)}$, for an approximation parameter $\eps>0$, by outputting a $(1+\eps)$-approximation to $\min_{\bx\in\mathbb{R}^d}\|\bA\bx-\bb^{(i)}\|_2$ for each query $\bb^{(i)}\in\mathbb{R}^n$ with $i\in[Q]$. 
For linear regression, we can again compute a subspace embedding $\bM=\bS\bA\in\mathbb{R}^{m\times n}$ and answer a query $\bb^{(i)}$ by approximately solving $\min_{\bx\in\mathbb{R}^d}\|\bS\bA\bx-\bS\bb^{(i)}\|_2$, where $\bS$ is a sketching matrix~\cite{ClarksonW13}. 

\begin{theorem}[\cite{ClarksonW13}]
\thmlab{thm:se}
Given $\bA\in\mathbb{R}^{n\times d}$, $\bb\in\mathbb{R}^n$, and an accuracy parameter $\eps>0$, there exists an algorithm that creates a data structure that uses $\O{\frac{d^2}{\eps^2}\log^2(nQ)}$ bits of space and outputs a $(1+\eps)$-approximation to $\min_{\bx\in\mathbb{R}^d}\|\bA\bx-\bb\|_2$ with high probability. 
\end{theorem}

However, this may fail for multiple interactions with the data structure. 
For example, suppose the adversary learns the kernel of $\bS$. 
Then the adversary could query some vector $\bb^{(i)}$ in the kernel of $\bS$ so that $\bS\bb^{(i)}$ is the all zeros vector, so that the output is the all zeros vector of dimension $d$, which could be arbitrarily bad compared to the actual minimizer. 
Thus the na\"{i}ve approach is to maintain $Q$ subspace embeddings, one for each query, resulting in a data structure with space $\tO{\frac{Qd}{\eps^2}}$. 
By comparison, \thmref{thm:se} and \thmref{thm:main:framework} yield the following:

\begin{theorem}
Given $\bA\in\mathbb{R}^{n\times d}$ and an accuracy parameter $\eps>0$, there exists an algorithm that creates a data structure that uses $\O{\frac{\sqrt{Q}d^2}{\eps^2}\log^3(nQ)}$ bits of space and with high probability, outputs $(1+\eps)$-approximations to $\min_{\bx\in\mathbb{R}^d}\|\bA\bx-\bb^{(i)}\|_2$ for $Q$ adaptive queries $\bb^{(1)},\ldots,\bb^{(Q)}$. 
\end{theorem}

\subsection{Application: Half-Space Queries}
Given a set $P$ of $n$ points in $\R^d$, the range query or search problem asks us to pre-process $P$ so that given a region $R$, chosen from a predetermined family, one can quickly count or return the points in $P \cap R$. 
This is an extremely well-studied class of problems in computational geometry \cite{toth2017handbook} and the case where the regions $R$ are hyperplanes (also called half-spaces) is of special interest since many algebraic constraints can be ``lifted" to be hyperplanes in a higher dimension. 

Unfortunately, exact versions of the problem are known to have the ``curse of dimensionality" and suffer from exponential dependence on $d$ in the query time \cite{bronnimann1993hard, chazelle2000discrepancy}. 
Nonetheless, \cite{chazelle2008approximate} gave a data structure capable of answering hyperplane queries \emph{approximately} with polynomial query time. 
Their notion of approximation is as follows: given a set of points $P$ in the unit $\ell_2$ ball, hyperplane $R$, and $\eps > 0$, we return the number of points that are on a given side of the hyperplane $R$ up to additive error equal to the number of points in $P$ which lie within distance $\eps$ of the boundary of $R$. We will refer to this query as an $\eps$-approximate hyperplane query. 
\cite{chazelle2008approximate} proved the following theorem.

\begin{theorem}[\cite{chazelle2008approximate}]
Given a set of points $P$ that lie in the unit $\ell_2$ ball, there exists a data structure that pre-processes $P$ using space $\tO{dn^{\O{\eps^{-2}}}}$ such that any $\eps$-approximate hyperplane range query is answered correctly with high probability. 
The query time is $\tO{d/\eps^2}$.
\end{theorem}

The data structure of \cite{chazelle2008approximate} is randomized and in particular employs randomized dimensionality reduction. Thus, it is feasible that queries might fail for multiple adaptive interactions with the data structure. By utilizing our framework of \secref{sec:framework} and \thmref{thm:main:framework}, we can obtain the following robust guarantee.

\begin{theorem}
Given a set of points $P$ that lie in the unit $\ell_2$ ball, there exists a data structure which pre-processes $P$ using space $\tO{\sqrt{Q}dn^{\O{\eps^{-2}}}}$ such that $Q$ adaptive $\eps$-hyperplane range queries are answered correctly with high probability. The query time is $\tO{d/\eps^2}$.
\end{theorem}

\subsection{Application: Point Queries on Turnstile Streams}
In the problem of point queries on turnstile streams, there exists a stream of $m$ updates. 
Each update specifies a coordinate $i\in[n]$ of an underlying frequency vector $f\in\mathbb{R}^n$ and changes $f_i$ by some amount between $\Delta_i\in[-\Delta,\Delta]$, where $\Delta=\poly(n)$. 
Given any \emph{constant} accuracy parameter $\eps>0$ any time $t\in[m]$, we define $f^{(t)}$ to be the frequency vector implicitly defined after the first $t$ updates. 
Then the point query problem is to output $f^{(t)}_i$ for various choices of $t\in[m]$ and $i\in[n]$ within an additive error of  $\eps\|f^{(t)}\|_1$.  

\begin{theorem}[\cite{AlmanY20}]
\thmlab{thm:turnstile:point}
There exists an algorithm that uses space $\O{\log^2 n}$ bits, worst-case update time $\O{\log^{0.582}n}$, and query time $\O{\log^{1.582}n}$, that supports point queries with $\eps=0.1$ with high probability. 
\end{theorem}
An important quality of \thmref{thm:turnstile:point} is that significantly improves the update time over previous data structures, e.g., \cite{CharikarCF04}, at a cost in query time. 
By applying \thmref{thm:main:framework}, we can avoid a blow-up in query time while still utilizing the update time improvements:
\begin{theorem}
There exists an algorithm that uses space $\O{\sqrt{Q}\log^3(nQ)}$ bits, has worst-case update time $\O{\sqrt{Q}\log^{1.582}(nQ)}$ and query time $\tO{\log^3(nQ)}$, and supports $Q$ adaptive point queries with $\eps=0.1$ and with high probability. 
\end{theorem}

\subsection{Adaptive Distance Estimation}
For completeness, we now show correctness of our algorithm across all $Q$ adaptive queries, though we remark that the proof can simply be black-boxed into \thmref{thm:main:framework}. 
\thmaderobust*
\begin{proof}
Fix query $(\by_q,i_q)$ with $q\in[Q]$ and $i_q\in[n]$. 
Let $S$ be a set of $k$ indices sampled (with replacement) from $[r]$. 
By \thmref{thm:jl} or \thmref{thm:fast:jl}, then we have for each $j\in[k]$,
\[\PPr{(1-\eps)\|\bx_{i_q}-\by_q\|_2\le\|\Pi_{S_j}(\bx_{i_q}-\by_q)\|_2\le(1+\eps)\|\bx_{i_q}-\by_q\|_2}\ge\frac{3}{4}.\]
Let $I_j$ be an indicator variable so that $I_j=1$ if $(1-\eps)\|\bx_{i_q}-\by_q\|_2\le\|\Pi_{S_j}(\bx_{i_q}-\by_q)\|_2\le(1+\eps)\|\bx_{i_q}-\by_q\|_2$ and $I_j=0$ otherwise, so that we have $\PPr{I_j=1}\ge\frac{3}{4}$, or equivalently, $\Ex{I_j}\ge\frac{3}{4}$. 
Let $I=\frac{1}{k}\sum_{j\in[k]}I_j$ so that by linearity of expectation, $\Ex{I}=\frac{1}{k}\sum_{j\in[k]}\Ex{I_j}\ge\frac{3}{4}$. 

To address adaptive queries, we first note that $\privmed$ is $(1,0)$-differentially private on the outputs of the $r$ Fast JL transforms. 
Since we sample $k=\O{\log(nQ)}$ groups from the $r=\O{\sqrt{Q}\log^2(nQ)}$ groups with replacement, then by amplification via sampling, i.e., \thmref{thm:dp:sampling}, $\privmed$ is $\left(\O{\frac{1}{\sqrt{Q}\log(nQ)}},0\right)$-differentially private. 
Thus, by the advanced composition of differential privacy, i.e., \thmref{thm:adaptive:queries}, the mechanism permits $Q$ adaptive queries and is $\left(\O{1},\frac{1}{\poly(nQ)}\right)$-differentially private. 
By the generalization properties of differential privacy, i.e., \thmref{thm:generalization}, we have
\[\PPr{\left|\frac{1}{k}\sum_{j\in[k]} I_j-\Ex{I}\right|\ge\frac{1}{10}}<\frac{1}{\poly(Q,n)},\]
for sufficiently small $\O{1}$. 
Thus we have
\[\PPr{\frac{1}{k}\sum_{i\in[k]} I_i>0.6}>1-\frac{1}{\poly(Q,n)},\]
which implies that $(1-\eps)\|\bx_{i_q}-\by_q\|_2\le d_i\le(1+\eps)\|\bx_{i_q}-\by_q\|_2$. 
Therefore, by a union bound across $Q$ adaptive queries $(\by_q,\bx_{i_q})$ with $q\in[Q]$, we have that
$(1-\eps)\|\bx_{i_q}-\by_q\|_2\le d_i\le(1+\eps)\|\bx_{i_q}-\by_q\|_2$ for all $q\in[Q]$ with high probability. 
\end{proof}

We similarly offer the following structural properties for \algref{ade_srht}. 
\begin{lemma}
\lemlab{lem:ade_single_iter_dp}
For all $i \in [n], j \in [Q]$, $T_j$ is $\left(o\lprp{\frac{1}{\sqrt{Q}\log(nQ)}}, 0\right)$-differentially private with $R_i$. 
\end{lemma}
\begin{proof}
    The proof is identical to that of \lemref{lem:dp:single:iter} with the observation that each transaction $T_j$ only results in a single query to a differentially private mechanism operating on $R_i$.
\end{proof}

\begin{lemma}
\lemlab{lem:ade_all_iter_dp}
For all $i \in [n]$, $T$ is $\left(o(1), \frac{1}{\poly(nQ)}\right)$-differentially private with respect to $R_i$. 
\end{lemma}
\begin{proof}
    The proof is identical to \lemref{lem:dp:all:iter} and follows from \thmref{thm:adaptive:queries} and \lemref{lem:ade_single_iter_dp}. 
\end{proof}

\subsection{Adaptive Kernel Density Estimation}
For completeness, we now show adversarial robustness of our algorithm across $Q$ adaptive queries. 
Again we remark that the proof can simply be black-boxed into \thmref{thm:main:framework}, though we include the specific kernel density details in the following proof as a warm-up for the following section.  
\lemkderobust*
\begin{proof}
Fix query $\by_q\in\mathbb{R}^d$  with $q\in[Q]$. 
Let $S$ be a set of $k$ indices sampled (with replacement) from $[r]$. 
Then by \thmref{thm:kde:better}, we have that for each $j\in[k]$,
\[\PPr{\left|D_{S_j}(\by)-\kde(X,\by)\right|\le\eps\cdot\kde(X,\by)}\ge\frac{3}{4}.\]
Let $I_j$ be an indicator variable so that $I_j=1$ if $\left|D_{S_j}(\by)-\kde(X,\by)\right|\le\eps\cdot\kde(X,\by)$ and $I_j=0$ otherwise, so that we have $\PPr{I_j=1}\ge\frac{3}{4}$ or equivalently, $\Ex{I_j}\ge\frac{3}{4}$. 
Let $I=\frac{1}{k}\sum_{j\in[k]} I_j$ so that $\Ex{I}=\frac{1}{k}\sum_{j\in[k]} \Ex{I_j}\ge\frac{3}{4}$. 
%

To handle adaptive queries, we first note that $\privmed$ is $(1,0)$-differentially private on the outputs of the $r$ kernel density estimation data structures. 
We sample $k=\O{\log Q}$ indices from the $r=\O{\sqrt{Q}\log^2 Q}$ data structures with replacement. 
Thus by amplification via sampling, i.e., \thmref{thm:dp:sampling}, $\privmed$ is $\left(\O{\frac{1}{\sqrt{Q}\log Q}},0\right)$-differentially private. 
By the advanced composition of differential privacy, i.e., \thmref{thm:adaptive:queries}, our algorithm can answer $Q$ adaptive queries with $\left(\O{1},\frac{1}{\poly(Q)}\right)$-differentially privacy. 
By the generalization properties of differential privacy, i.e., \thmref{thm:generalization}, we have
\[\PPr{\left|\frac{1}{k}\sum_{j\in[k]} I_j-\Ex{I}\right|\ge\frac{1}{10}}<0.01,\]
for sufficiently small constant $\O{1}$ in the private median algorithm $\privmed$. 
Therefore,
\[\PPr{\frac{1}{k}\sum_{i\in[k]} I_i>0.6}>0.99,\]
so that 
$\left|d_q-\kde(X,\by_q)\right|\le\eps\cdot\kde(X,\by)$ across $Q$ queries $\by_q$ with $q\in[Q]$.
\end{proof}

\end{document}